\documentclass{article}

\usepackage{lineno,hyperref}
\usepackage{amsthm}

\modulolinenumbers[5]

\usepackage[latin1]{inputenc}
\usepackage{subfigure}
\usepackage{color}
\usepackage{epic}
\usepackage{tikz}

\definecolor{AliceBlue}{rgb}{0.94,0.97,1.00}
\definecolor{AntiqueWhite1}{rgb}{1.00,0.94,0.86}
\definecolor{AntiqueWhite2}{rgb}{0.93,0.87,0.80}
\definecolor{AntiqueWhite3}{rgb}{0.80,0.75,0.69}
\definecolor{AntiqueWhite4}{rgb}{0.55,0.51,0.47}
\definecolor{AntiqueWhite}{rgb}{0.98,0.92,0.84}
\definecolor{BlanchedAlmond}{rgb}{1.00,0.92,0.80}
\definecolor{BlueViolet}{rgb}{0.54,0.17,0.89}
\definecolor{CadetBlue1}{rgb}{0.60,0.96,1.00}
\definecolor{CadetBlue2}{rgb}{0.56,0.90,0.93}
\definecolor{CadetBlue3}{rgb}{0.48,0.77,0.80}
\definecolor{CadetBlue4}{rgb}{0.33,0.53,0.55}
\definecolor{CadetBlue}{rgb}{0.37,0.62,0.63}
\definecolor{CornflowerBlue}{rgb}{0.39,0.58,0.93}
\definecolor{DarkBlue}{rgb}{0.00,0.00,0.55}
\definecolor{DarkCyan}{rgb}{0.00,0.55,0.55}
\definecolor{DarkGoldenrod1}{rgb}{1.00,0.73,0.06}
\definecolor{DarkGoldenrod2}{rgb}{0.93,0.68,0.05}
\definecolor{DarkGoldenrod3}{rgb}{0.80,0.58,0.05}
\definecolor{DarkGoldenrod4}{rgb}{0.55,0.40,0.03}
\definecolor{DarkGoldenrod}{rgb}{0.72,0.53,0.04}
\definecolor{DarkGray}{rgb}{0.66,0.66,0.66}
\definecolor{DarkGreen}{rgb}{0.00,0.39,0.00}
\definecolor{DarkGrey}{rgb}{0.66,0.66,0.66}
\definecolor{DarkKhaki}{rgb}{0.74,0.72,0.42}
\definecolor{DarkMagenta}{rgb}{0.55,0.00,0.55}
\definecolor{DarkOliveGreen1}{rgb}{0.79,1.00,0.44}
\definecolor{DarkOliveGreen2}{rgb}{0.74,0.93,0.41}
\definecolor{DarkOliveGreen3}{rgb}{0.64,0.80,0.35}
\definecolor{DarkOliveGreen4}{rgb}{0.43,0.55,0.24}
\definecolor{DarkOliveGreen}{rgb}{0.33,0.42,0.18}
\definecolor{DarkOrange1}{rgb}{1.00,0.50,0.00}
\definecolor{DarkOrange2}{rgb}{0.93,0.46,0.00}
\definecolor{DarkOrange3}{rgb}{0.80,0.40,0.00}
\definecolor{DarkOrange4}{rgb}{0.55,0.27,0.00}
\definecolor{DarkOrange}{rgb}{1.00,0.55,0.00}
\definecolor{DarkOrchid1}{rgb}{0.75,0.24,1.00}
\definecolor{DarkOrchid2}{rgb}{0.70,0.23,0.93}
\definecolor{DarkOrchid3}{rgb}{0.60,0.20,0.80}
\definecolor{DarkOrchid4}{rgb}{0.41,0.13,0.55}
\definecolor{DarkOrchid}{rgb}{0.60,0.20,0.80}
\definecolor{DarkRed}{rgb}{0.55,0.00,0.00}
\definecolor{DarkSalmon}{rgb}{0.91,0.59,0.48}
\definecolor{DarkSeaGreen1}{rgb}{0.76,1.00,0.76}
\definecolor{DarkSeaGreen2}{rgb}{0.71,0.93,0.71}
\definecolor{DarkSeaGreen3}{rgb}{0.61,0.80,0.61}
\definecolor{DarkSeaGreen4}{rgb}{0.41,0.55,0.41}
\definecolor{DarkSeaGreen}{rgb}{0.56,0.74,0.56}
\definecolor{DarkSlateBlue}{rgb}{0.28,0.24,0.55}
\definecolor{DarkSlateGray1}{rgb}{0.59,1.00,1.00}
\definecolor{DarkSlateGray2}{rgb}{0.55,0.93,0.93}
\definecolor{DarkSlateGray3}{rgb}{0.47,0.80,0.80}
\definecolor{DarkSlateGray4}{rgb}{0.32,0.55,0.55}
\definecolor{DarkSlateGray}{rgb}{0.18,0.31,0.31}
\definecolor{DarkSlateGrey}{rgb}{0.18,0.31,0.31}
\definecolor{DarkTurquoise}{rgb}{0.00,0.81,0.82}
\definecolor{DarkViolet}{rgb}{0.58,0.00,0.83}
\definecolor{DeepPink1}{rgb}{1.00,0.08,0.58}
\definecolor{DeepPink2}{rgb}{0.93,0.07,0.54}
\definecolor{DeepPink3}{rgb}{0.80,0.06,0.46}
\definecolor{DeepPink4}{rgb}{0.55,0.04,0.31}
\definecolor{DeepPink}{rgb}{1.00,0.08,0.58}
\definecolor{DeepSkyBlue1}{rgb}{0.00,0.75,1.00}
\definecolor{DeepSkyBlue2}{rgb}{0.00,0.70,0.93}
\definecolor{DeepSkyBlue3}{rgb}{0.00,0.60,0.80}
\definecolor{DeepSkyBlue4}{rgb}{0.00,0.41,0.55}
\definecolor{DeepSkyBlue}{rgb}{0.00,0.75,1.00}
\definecolor{DimGray}{rgb}{0.41,0.41,0.41}
\definecolor{DimGrey}{rgb}{0.41,0.41,0.41}
\definecolor{DodgerBlue1}{rgb}{0.12,0.56,1.00}
\definecolor{DodgerBlue2}{rgb}{0.11,0.53,0.93}
\definecolor{DodgerBlue3}{rgb}{0.09,0.45,0.80}
\definecolor{DodgerBlue4}{rgb}{0.06,0.31,0.55}
\definecolor{DodgerBlue}{rgb}{0.12,0.56,1.00}
\definecolor{FloralWhite}{rgb}{1.00,0.98,0.94}
\definecolor{ForestGreen}{rgb}{0.13,0.55,0.13}
\definecolor{GhostWhite}{rgb}{0.97,0.97,1.00}
\definecolor{GreenYellow}{rgb}{0.68,1.00,0.18}
\definecolor{HotPink1}{rgb}{1.00,0.43,0.71}
\definecolor{HotPink2}{rgb}{0.93,0.42,0.65}
\definecolor{HotPink3}{rgb}{0.80,0.38,0.56}
\definecolor{HotPink4}{rgb}{0.55,0.23,0.38}
\definecolor{HotPink}{rgb}{1.00,0.41,0.71}
\definecolor{IndianRed1}{rgb}{1.00,0.42,0.42}
\definecolor{IndianRed2}{rgb}{0.93,0.39,0.39}
\definecolor{IndianRed3}{rgb}{0.80,0.33,0.33}
\definecolor{IndianRed4}{rgb}{0.55,0.23,0.23}
\definecolor{IndianRed}{rgb}{0.80,0.36,0.36}
\definecolor{LavenderBlush1}{rgb}{1.00,0.94,0.96}
\definecolor{LavenderBlush2}{rgb}{0.93,0.88,0.90}
\definecolor{LavenderBlush3}{rgb}{0.80,0.76,0.77}
\definecolor{LavenderBlush4}{rgb}{0.55,0.51,0.53}
\definecolor{LavenderBlush}{rgb}{1.00,0.94,0.96}
\definecolor{LawnGreen}{rgb}{0.49,0.99,0.00}
\definecolor{LemonChiffon1}{rgb}{1.00,0.98,0.80}
\definecolor{LemonChiffon2}{rgb}{0.93,0.91,0.75}
\definecolor{LemonChiffon3}{rgb}{0.80,0.79,0.65}
\definecolor{LemonChiffon4}{rgb}{0.55,0.54,0.44}
\definecolor{LemonChiffon}{rgb}{1.00,0.98,0.80}
\definecolor{LightBlue1}{rgb}{0.75,0.94,1.00}
\definecolor{LightBlue2}{rgb}{0.70,0.87,0.93}
\definecolor{LightBlue3}{rgb}{0.60,0.75,0.80}
\definecolor{LightBlue4}{rgb}{0.41,0.51,0.55}
\definecolor{LightBlue}{rgb}{0.68,0.85,0.90}
\definecolor{LightCoral}{rgb}{0.94,0.50,0.50}
\definecolor{LightCyan1}{rgb}{0.88,1.00,1.00}
\definecolor{LightCyan2}{rgb}{0.82,0.93,0.93}
\definecolor{LightCyan3}{rgb}{0.71,0.80,0.80}
\definecolor{LightCyan4}{rgb}{0.48,0.55,0.55}
\definecolor{LightCyan}{rgb}{0.88,1.00,1.00}
\definecolor{LightGoldenrod1}{rgb}{1.00,0.93,0.55}
\definecolor{LightGoldenrod2}{rgb}{0.93,0.86,0.51}
\definecolor{LightGoldenrod3}{rgb}{0.80,0.75,0.44}
\definecolor{LightGoldenrod4}{rgb}{0.55,0.51,0.30}
\definecolor{LightGoldenrodYellow}{rgb}{0.98,0.98,0.82}
\definecolor{LightGoldenrod}{rgb}{0.93,0.87,0.51}
\definecolor{LightGray}{rgb}{0.83,0.83,0.83}
\definecolor{LightGreen}{rgb}{0.56,0.93,0.56}
\definecolor{LightGrey}{rgb}{0.83,0.83,0.83}
\definecolor{LightPink1}{rgb}{1.00,0.68,0.73}
\definecolor{LightPink2}{rgb}{0.93,0.64,0.68}
\definecolor{LightPink3}{rgb}{0.80,0.55,0.58}
\definecolor{LightPink4}{rgb}{0.55,0.37,0.40}
\definecolor{LightPink}{rgb}{1.00,0.71,0.76}
\definecolor{LightSalmon1}{rgb}{1.00,0.63,0.48}
\definecolor{LightSalmon2}{rgb}{0.93,0.58,0.45}
\definecolor{LightSalmon3}{rgb}{0.80,0.51,0.38}
\definecolor{LightSalmon4}{rgb}{0.55,0.34,0.26}
\definecolor{LightSalmon}{rgb}{1.00,0.63,0.48}
\definecolor{LightSeaGreen}{rgb}{0.13,0.70,0.67}
\definecolor{LightSkyBlue1}{rgb}{0.69,0.89,1.00}
\definecolor{LightSkyBlue2}{rgb}{0.64,0.83,0.93}
\definecolor{LightSkyBlue3}{rgb}{0.55,0.71,0.80}
\definecolor{LightSkyBlue4}{rgb}{0.38,0.48,0.55}
\definecolor{LightSkyBlue}{rgb}{0.53,0.81,0.98}
\definecolor{LightSlateBlue}{rgb}{0.52,0.44,1.00}
\definecolor{LightSlateGray}{rgb}{0.47,0.53,0.60}
\definecolor{LightSlateGrey}{rgb}{0.47,0.53,0.60}
\definecolor{LightSteelBlue1}{rgb}{0.79,0.88,1.00}
\definecolor{LightSteelBlue2}{rgb}{0.74,0.82,0.93}
\definecolor{LightSteelBlue3}{rgb}{0.64,0.71,0.80}
\definecolor{LightSteelBlue4}{rgb}{0.43,0.48,0.55}
\definecolor{LightSteelBlue}{rgb}{0.69,0.77,0.87}
\definecolor{LightYellow1}{rgb}{1.00,1.00,0.88}
\definecolor{LightYellow2}{rgb}{0.93,0.93,0.82}
\definecolor{LightYellow3}{rgb}{0.80,0.80,0.71}
\definecolor{LightYellow4}{rgb}{0.55,0.55,0.48}
\definecolor{LightYellow}{rgb}{1.00,1.00,0.88}
\definecolor{LimeGreen}{rgb}{0.20,0.80,0.20}
\definecolor{MediumAquamarine}{rgb}{0.40,0.80,0.67}
\definecolor{MediumBlue}{rgb}{0.00,0.00,0.80}
\definecolor{MediumOrchid1}{rgb}{0.88,0.40,1.00}
\definecolor{MediumOrchid2}{rgb}{0.82,0.37,0.93}
\definecolor{MediumOrchid3}{rgb}{0.71,0.32,0.80}
\definecolor{MediumOrchid4}{rgb}{0.48,0.22,0.55}
\definecolor{MediumOrchid}{rgb}{0.73,0.33,0.83}
\definecolor{MediumPurple1}{rgb}{0.67,0.51,1.00}
\definecolor{MediumPurple2}{rgb}{0.62,0.47,0.93}
\definecolor{MediumPurple3}{rgb}{0.54,0.41,0.80}
\definecolor{MediumPurple4}{rgb}{0.36,0.28,0.55}
\definecolor{MediumPurple}{rgb}{0.58,0.44,0.86}
\definecolor{MediumSeaGreen}{rgb}{0.24,0.70,0.44}
\definecolor{MediumSlateBlue}{rgb}{0.48,0.41,0.93}
\definecolor{MediumSpringGreen}{rgb}{0.00,0.98,0.60}
\definecolor{MediumTurquoise}{rgb}{0.28,0.82,0.80}
\definecolor{MediumVioletRed}{rgb}{0.78,0.08,0.52}
\definecolor{MidnightBlue}{rgb}{0.10,0.10,0.44}
\definecolor{MintCream}{rgb}{0.96,1.00,0.98}
\definecolor{MistyRose1}{rgb}{1.00,0.89,0.88}
\definecolor{MistyRose2}{rgb}{0.93,0.84,0.82}
\definecolor{MistyRose3}{rgb}{0.80,0.72,0.71}
\definecolor{MistyRose4}{rgb}{0.55,0.49,0.48}
\definecolor{MistyRose}{rgb}{1.00,0.89,0.88}
\definecolor{NavajoWhite1}{rgb}{1.00,0.87,0.68}
\definecolor{NavajoWhite2}{rgb}{0.93,0.81,0.63}
\definecolor{NavajoWhite3}{rgb}{0.80,0.70,0.55}
\definecolor{NavajoWhite4}{rgb}{0.55,0.47,0.37}
\definecolor{NavajoWhite}{rgb}{1.00,0.87,0.68}
\definecolor{NavyBlue}{rgb}{0.00,0.00,0.50}
\definecolor{OldLace}{rgb}{0.99,0.96,0.90}
\definecolor{OliveDrab1}{rgb}{0.75,1.00,0.24}
\definecolor{OliveDrab2}{rgb}{0.70,0.93,0.23}
\definecolor{OliveDrab3}{rgb}{0.60,0.80,0.20}
\definecolor{OliveDrab4}{rgb}{0.41,0.55,0.13}
\definecolor{OliveDrab}{rgb}{0.42,0.56,0.14}
\definecolor{OrangeRed1}{rgb}{1.00,0.27,0.00}
\definecolor{OrangeRed2}{rgb}{0.93,0.25,0.00}
\definecolor{OrangeRed3}{rgb}{0.80,0.22,0.00}
\definecolor{OrangeRed4}{rgb}{0.55,0.15,0.00}
\definecolor{OrangeRed}{rgb}{1.00,0.27,0.00}
\definecolor{PaleGoldenrod}{rgb}{0.93,0.91,0.67}
\definecolor{PaleGreen1}{rgb}{0.60,1.00,0.60}
\definecolor{PaleGreen2}{rgb}{0.56,0.93,0.56}
\definecolor{PaleGreen3}{rgb}{0.49,0.80,0.49}
\definecolor{PaleGreen4}{rgb}{0.33,0.55,0.33}
\definecolor{PaleGreen}{rgb}{0.60,0.98,0.60}
\definecolor{PaleTurquoise1}{rgb}{0.73,1.00,1.00}
\definecolor{PaleTurquoise2}{rgb}{0.68,0.93,0.93}
\definecolor{PaleTurquoise3}{rgb}{0.59,0.80,0.80}
\definecolor{PaleTurquoise4}{rgb}{0.40,0.55,0.55}
\definecolor{PaleTurquoise}{rgb}{0.69,0.93,0.93}
\definecolor{PaleVioletRed1}{rgb}{1.00,0.51,0.67}
\definecolor{PaleVioletRed2}{rgb}{0.93,0.47,0.62}
\definecolor{PaleVioletRed3}{rgb}{0.80,0.41,0.54}
\definecolor{PaleVioletRed4}{rgb}{0.55,0.28,0.36}
\definecolor{PaleVioletRed}{rgb}{0.86,0.44,0.58}
\definecolor{PapayaWhip}{rgb}{1.00,0.94,0.84}
\definecolor{PeachPuff1}{rgb}{1.00,0.85,0.73}
\definecolor{PeachPuff2}{rgb}{0.93,0.80,0.68}
\definecolor{PeachPuff3}{rgb}{0.80,0.69,0.58}
\definecolor{PeachPuff4}{rgb}{0.55,0.47,0.40}
\definecolor{PeachPuff}{rgb}{1.00,0.85,0.73}
\definecolor{PowderBlue}{rgb}{0.69,0.88,0.90}
\definecolor{RosyBrown1}{rgb}{1.00,0.76,0.76}
\definecolor{RosyBrown2}{rgb}{0.93,0.71,0.71}
\definecolor{RosyBrown3}{rgb}{0.80,0.61,0.61}
\definecolor{RosyBrown4}{rgb}{0.55,0.41,0.41}
\definecolor{RosyBrown}{rgb}{0.74,0.56,0.56}
\definecolor{RoyalBlue1}{rgb}{0.28,0.46,1.00}
\definecolor{RoyalBlue2}{rgb}{0.26,0.43,0.93}
\definecolor{RoyalBlue3}{rgb}{0.23,0.37,0.80}
\definecolor{RoyalBlue4}{rgb}{0.15,0.25,0.55}
\definecolor{RoyalBlue}{rgb}{0.25,0.41,0.88}
\definecolor{SaddleBrown}{rgb}{0.55,0.27,0.07}
\definecolor{SandyBrown}{rgb}{0.96,0.64,0.38}
\definecolor{SeaGreen1}{rgb}{0.33,1.00,0.62}
\definecolor{SeaGreen2}{rgb}{0.31,0.93,0.58}
\definecolor{SeaGreen3}{rgb}{0.26,0.80,0.50}
\definecolor{SeaGreen4}{rgb}{0.18,0.55,0.34}
\definecolor{SeaGreen}{rgb}{0.18,0.55,0.34}
\definecolor{SkyBlue1}{rgb}{0.53,0.81,1.00}
\definecolor{SkyBlue2}{rgb}{0.49,0.75,0.93}
\definecolor{SkyBlue3}{rgb}{0.42,0.65,0.80}
\definecolor{SkyBlue4}{rgb}{0.29,0.44,0.55}
\definecolor{SkyBlue}{rgb}{0.53,0.81,0.92}
\definecolor{SlateBlue1}{rgb}{0.51,0.44,1.00}
\definecolor{SlateBlue2}{rgb}{0.48,0.40,0.93}
\definecolor{SlateBlue3}{rgb}{0.41,0.35,0.80}
\definecolor{SlateBlue4}{rgb}{0.28,0.24,0.55}
\definecolor{SlateBlue}{rgb}{0.42,0.35,0.80}
\definecolor{SlateGray1}{rgb}{0.78,0.89,1.00}
\definecolor{SlateGray2}{rgb}{0.73,0.83,0.93}
\definecolor{SlateGray3}{rgb}{0.62,0.71,0.80}
\definecolor{SlateGray4}{rgb}{0.42,0.48,0.55}
\definecolor{SlateGray}{rgb}{0.44,0.50,0.56}
\definecolor{SlateGrey}{rgb}{0.44,0.50,0.56}
\definecolor{SpringGreen1}{rgb}{0.00,1.00,0.50}
\definecolor{SpringGreen2}{rgb}{0.00,0.93,0.46}
\definecolor{SpringGreen3}{rgb}{0.00,0.80,0.40}
\definecolor{SpringGreen4}{rgb}{0.00,0.55,0.27}
\definecolor{SpringGreen}{rgb}{0.00,1.00,0.50}
\definecolor{SteelBlue1}{rgb}{0.39,0.72,1.00}
\definecolor{SteelBlue2}{rgb}{0.36,0.67,0.93}
\definecolor{SteelBlue3}{rgb}{0.31,0.58,0.80}
\definecolor{SteelBlue4}{rgb}{0.21,0.39,0.55}
\definecolor{SteelBlue}{rgb}{0.27,0.51,0.71}
\definecolor{VioletRed1}{rgb}{1.00,0.24,0.59}
\definecolor{VioletRed2}{rgb}{0.93,0.23,0.55}
\definecolor{VioletRed3}{rgb}{0.80,0.20,0.47}
\definecolor{VioletRed4}{rgb}{0.55,0.13,0.32}
\definecolor{VioletRed}{rgb}{0.82,0.13,0.56}
\definecolor{WhiteSmoke}{rgb}{0.96,0.96,0.96}
\definecolor{YellowGreen}{rgb}{0.60,0.80,0.20}
\definecolor{aliceblue}{rgb}{0.94,0.97,1.00}
\definecolor{antiquewhite}{rgb}{0.98,0.92,0.84}
\definecolor{aquamarine1}{rgb}{0.50,1.00,0.83}
\definecolor{aquamarine2}{rgb}{0.46,0.93,0.78}
\definecolor{aquamarine3}{rgb}{0.40,0.80,0.67}
\definecolor{aquamarine4}{rgb}{0.27,0.55,0.45}
\definecolor{aquamarine}{rgb}{0.50,1.00,0.83}
\definecolor{azure1}{rgb}{0.94,1.00,1.00}
\definecolor{azure2}{rgb}{0.88,0.93,0.93}
\definecolor{azure3}{rgb}{0.76,0.80,0.80}
\definecolor{azure4}{rgb}{0.51,0.55,0.55}
\definecolor{azure}{rgb}{0.94,1.00,1.00}
\definecolor{beige}{rgb}{0.96,0.96,0.86}
\definecolor{bisque1}{rgb}{1.00,0.89,0.77}
\definecolor{bisque2}{rgb}{0.93,0.84,0.72}
\definecolor{bisque3}{rgb}{0.80,0.72,0.62}
\definecolor{bisque4}{rgb}{0.55,0.49,0.42}
\definecolor{bisque}{rgb}{1.00,0.89,0.77}
\definecolor{black}{rgb}{0.00,0.00,0.00}
\definecolor{blanchedalmond}{rgb}{1.00,0.92,0.80}
\definecolor{blue1}{rgb}{0.00,0.00,1.00}
\definecolor{blue2}{rgb}{0.00,0.00,0.93}
\definecolor{blue3}{rgb}{0.00,0.00,0.80}
\definecolor{blue4}{rgb}{0.00,0.00,0.55}
\definecolor{blueviolet}{rgb}{0.54,0.17,0.89}
\definecolor{blue}{rgb}{0.00,0.00,1.00}
\definecolor{brown1}{rgb}{1.00,0.25,0.25}
\definecolor{brown2}{rgb}{0.93,0.23,0.23}
\definecolor{brown3}{rgb}{0.80,0.20,0.20}
\definecolor{brown4}{rgb}{0.55,0.14,0.14}
\definecolor{brown}{rgb}{0.65,0.16,0.16}
\definecolor{burlywood1}{rgb}{1.00,0.83,0.61}
\definecolor{burlywood2}{rgb}{0.93,0.77,0.57}
\definecolor{burlywood3}{rgb}{0.80,0.67,0.49}
\definecolor{burlywood4}{rgb}{0.55,0.45,0.33}
\definecolor{burlywood}{rgb}{0.87,0.72,0.53}
\definecolor{cadetblue}{rgb}{0.37,0.62,0.63}
\definecolor{chartreuse1}{rgb}{0.50,1.00,0.00}
\definecolor{chartreuse2}{rgb}{0.46,0.93,0.00}
\definecolor{chartreuse3}{rgb}{0.40,0.80,0.00}
\definecolor{chartreuse4}{rgb}{0.27,0.55,0.00}
\definecolor{chartreuse}{rgb}{0.50,1.00,0.00}
\definecolor{chocolate1}{rgb}{1.00,0.50,0.14}
\definecolor{chocolate2}{rgb}{0.93,0.46,0.13}
\definecolor{chocolate3}{rgb}{0.80,0.40,0.11}
\definecolor{chocolate4}{rgb}{0.55,0.27,0.07}
\definecolor{chocolate}{rgb}{0.82,0.41,0.12}
\definecolor{coral1}{rgb}{1.00,0.45,0.34}
\definecolor{coral2}{rgb}{0.93,0.42,0.31}
\definecolor{coral3}{rgb}{0.80,0.36,0.27}
\definecolor{coral4}{rgb}{0.55,0.24,0.18}
\definecolor{coral}{rgb}{1.00,0.50,0.31}
\definecolor{cornflowerblue}{rgb}{0.39,0.58,0.93}
\definecolor{cornsilk1}{rgb}{1.00,0.97,0.86}
\definecolor{cornsilk2}{rgb}{0.93,0.91,0.80}
\definecolor{cornsilk3}{rgb}{0.80,0.78,0.69}
\definecolor{cornsilk4}{rgb}{0.55,0.53,0.47}
\definecolor{cornsilk}{rgb}{1.00,0.97,0.86}
\definecolor{cyan1}{rgb}{0.00,1.00,1.00}
\definecolor{cyan2}{rgb}{0.00,0.93,0.93}
\definecolor{cyan3}{rgb}{0.00,0.80,0.80}
\definecolor{cyan4}{rgb}{0.00,0.55,0.55}
\definecolor{cyan}{rgb}{0.00,1.00,1.00}
\definecolor{darkblue}{rgb}{0.00,0.00,0.55}
\definecolor{darkcyan}{rgb}{0.00,0.55,0.55}
\definecolor{darkgoldenrod}{rgb}{0.72,0.53,0.04}
\definecolor{darkgray}{rgb}{0.66,0.66,0.66}
\definecolor{darkgreen}{rgb}{0.00,0.39,0.00}
\definecolor{darkgrey}{rgb}{0.66,0.66,0.66}
\definecolor{darkkhaki}{rgb}{0.74,0.72,0.42}
\definecolor{darkmagenta}{rgb}{0.55,0.00,0.55}
\definecolor{darkolive}{rgb}{0.33,0.42,0.18}
\definecolor{darkorange}{rgb}{1.00,0.55,0.00}
\definecolor{darkorchid}{rgb}{0.60,0.20,0.80}
\definecolor{darkred}{rgb}{0.55,0.00,0.00}
\definecolor{darksalmon}{rgb}{0.91,0.59,0.48}
\definecolor{darksea}{rgb}{0.56,0.74,0.56}
\definecolor{darkslate}{rgb}{0.18,0.31,0.31}
\definecolor{darkslate}{rgb}{0.18,0.31,0.31}
\definecolor{darkslate}{rgb}{0.28,0.24,0.55}
\definecolor{darkturquoise}{rgb}{0.00,0.81,0.82}
\definecolor{darkviolet}{rgb}{0.58,0.00,0.83}
\definecolor{deeppink}{rgb}{1.00,0.08,0.58}
\definecolor{deepsky}{rgb}{0.00,0.75,1.00}
\definecolor{dimgray}{rgb}{0.41,0.41,0.41}
\definecolor{dimgrey}{rgb}{0.41,0.41,0.41}
\definecolor{dodgerblue}{rgb}{0.12,0.56,1.00}
\definecolor{firebrick1}{rgb}{1.00,0.19,0.19}
\definecolor{firebrick2}{rgb}{0.93,0.17,0.17}
\definecolor{firebrick3}{rgb}{0.80,0.15,0.15}
\definecolor{firebrick4}{rgb}{0.55,0.10,0.10}
\definecolor{firebrick}{rgb}{0.70,0.13,0.13}
\definecolor{floralwhite}{rgb}{1.00,0.98,0.94}
\definecolor{forestgreen}{rgb}{0.13,0.55,0.13}
\definecolor{gainsboro}{rgb}{0.86,0.86,0.86}
\definecolor{ghostwhite}{rgb}{0.97,0.97,1.00}
\definecolor{gold1}{rgb}{1.00,0.84,0.00}
\definecolor{gold2}{rgb}{0.93,0.79,0.00}
\definecolor{gold3}{rgb}{0.80,0.68,0.00}
\definecolor{gold4}{rgb}{0.55,0.46,0.00}
\definecolor{goldenrod1}{rgb}{1.00,0.76,0.15}
\definecolor{goldenrod2}{rgb}{0.93,0.71,0.13}
\definecolor{goldenrod3}{rgb}{0.80,0.61,0.11}
\definecolor{goldenrod4}{rgb}{0.55,0.41,0.08}
\definecolor{goldenrod}{rgb}{0.85,0.65,0.13}
\definecolor{gold}{rgb}{1.00,0.84,0.00}
\definecolor{gray0}{rgb}{0.00,0.00,0.00}
\definecolor{gray100}{rgb}{1.00,1.00,1.00}
\definecolor{gray10}{rgb}{0.10,0.10,0.10}
\definecolor{gray11}{rgb}{0.11,0.11,0.11}
\definecolor{gray12}{rgb}{0.12,0.12,0.12}
\definecolor{gray13}{rgb}{0.13,0.13,0.13}
\definecolor{gray14}{rgb}{0.14,0.14,0.14}
\definecolor{gray15}{rgb}{0.15,0.15,0.15}
\definecolor{gray16}{rgb}{0.16,0.16,0.16}
\definecolor{gray17}{rgb}{0.17,0.17,0.17}
\definecolor{gray18}{rgb}{0.18,0.18,0.18}
\definecolor{gray19}{rgb}{0.19,0.19,0.19}
\definecolor{gray1}{rgb}{0.01,0.01,0.01}
\definecolor{gray20}{rgb}{0.20,0.20,0.20}
\definecolor{gray21}{rgb}{0.21,0.21,0.21}
\definecolor{gray22}{rgb}{0.22,0.22,0.22}
\definecolor{gray23}{rgb}{0.23,0.23,0.23}
\definecolor{gray24}{rgb}{0.24,0.24,0.24}
\definecolor{gray25}{rgb}{0.25,0.25,0.25}
\definecolor{gray26}{rgb}{0.26,0.26,0.26}
\definecolor{gray27}{rgb}{0.27,0.27,0.27}
\definecolor{gray28}{rgb}{0.28,0.28,0.28}
\definecolor{gray29}{rgb}{0.29,0.29,0.29}
\definecolor{gray2}{rgb}{0.02,0.02,0.02}
\definecolor{gray30}{rgb}{0.30,0.30,0.30}
\definecolor{gray31}{rgb}{0.31,0.31,0.31}
\definecolor{gray32}{rgb}{0.32,0.32,0.32}
\definecolor{gray33}{rgb}{0.33,0.33,0.33}
\definecolor{gray34}{rgb}{0.34,0.34,0.34}
\definecolor{gray35}{rgb}{0.35,0.35,0.35}
\definecolor{gray36}{rgb}{0.36,0.36,0.36}
\definecolor{gray37}{rgb}{0.37,0.37,0.37}
\definecolor{gray38}{rgb}{0.38,0.38,0.38}
\definecolor{gray39}{rgb}{0.39,0.39,0.39}
\definecolor{gray3}{rgb}{0.03,0.03,0.03}
\definecolor{gray40}{rgb}{0.40,0.40,0.40}
\definecolor{gray41}{rgb}{0.41,0.41,0.41}
\definecolor{gray42}{rgb}{0.42,0.42,0.42}
\definecolor{gray43}{rgb}{0.43,0.43,0.43}
\definecolor{gray44}{rgb}{0.44,0.44,0.44}
\definecolor{gray45}{rgb}{0.45,0.45,0.45}
\definecolor{gray46}{rgb}{0.46,0.46,0.46}
\definecolor{gray47}{rgb}{0.47,0.47,0.47}
\definecolor{gray48}{rgb}{0.48,0.48,0.48}
\definecolor{gray49}{rgb}{0.49,0.49,0.49}
\definecolor{gray4}{rgb}{0.04,0.04,0.04}
\definecolor{gray50}{rgb}{0.50,0.50,0.50}
\definecolor{gray51}{rgb}{0.51,0.51,0.51}
\definecolor{gray52}{rgb}{0.52,0.52,0.52}
\definecolor{gray53}{rgb}{0.53,0.53,0.53}
\definecolor{gray54}{rgb}{0.54,0.54,0.54}
\definecolor{gray55}{rgb}{0.55,0.55,0.55}
\definecolor{gray56}{rgb}{0.56,0.56,0.56}
\definecolor{gray57}{rgb}{0.57,0.57,0.57}
\definecolor{gray58}{rgb}{0.58,0.58,0.58}
\definecolor{gray59}{rgb}{0.59,0.59,0.59}
\definecolor{gray5}{rgb}{0.05,0.05,0.05}
\definecolor{gray60}{rgb}{0.60,0.60,0.60}
\definecolor{gray61}{rgb}{0.61,0.61,0.61}
\definecolor{gray62}{rgb}{0.62,0.62,0.62}
\definecolor{gray63}{rgb}{0.63,0.63,0.63}
\definecolor{gray64}{rgb}{0.64,0.64,0.64}
\definecolor{gray65}{rgb}{0.65,0.65,0.65}
\definecolor{gray66}{rgb}{0.66,0.66,0.66}
\definecolor{gray67}{rgb}{0.67,0.67,0.67}
\definecolor{gray68}{rgb}{0.68,0.68,0.68}
\definecolor{gray69}{rgb}{0.69,0.69,0.69}
\definecolor{gray6}{rgb}{0.06,0.06,0.06}
\definecolor{gray70}{rgb}{0.70,0.70,0.70}
\definecolor{gray71}{rgb}{0.71,0.71,0.71}
\definecolor{gray72}{rgb}{0.72,0.72,0.72}
\definecolor{gray73}{rgb}{0.73,0.73,0.73}
\definecolor{gray74}{rgb}{0.74,0.74,0.74}
\definecolor{gray75}{rgb}{0.75,0.75,0.75}
\definecolor{gray76}{rgb}{0.76,0.76,0.76}
\definecolor{gray77}{rgb}{0.77,0.77,0.77}
\definecolor{gray78}{rgb}{0.78,0.78,0.78}
\definecolor{gray79}{rgb}{0.79,0.79,0.79}
\definecolor{gray7}{rgb}{0.07,0.07,0.07}
\definecolor{gray80}{rgb}{0.80,0.80,0.80}
\definecolor{gray81}{rgb}{0.81,0.81,0.81}
\definecolor{gray82}{rgb}{0.82,0.82,0.82}
\definecolor{gray83}{rgb}{0.83,0.83,0.83}
\definecolor{gray84}{rgb}{0.84,0.84,0.84}
\definecolor{gray85}{rgb}{0.85,0.85,0.85}
\definecolor{gray86}{rgb}{0.86,0.86,0.86}
\definecolor{gray87}{rgb}{0.87,0.87,0.87}
\definecolor{gray88}{rgb}{0.88,0.88,0.88}
\definecolor{gray89}{rgb}{0.89,0.89,0.89}
\definecolor{gray8}{rgb}{0.08,0.08,0.08}
\definecolor{gray90}{rgb}{0.90,0.90,0.90}
\definecolor{gray91}{rgb}{0.91,0.91,0.91}
\definecolor{gray92}{rgb}{0.92,0.92,0.92}
\definecolor{gray93}{rgb}{0.93,0.93,0.93}
\definecolor{gray94}{rgb}{0.94,0.94,0.94}
\definecolor{gray95}{rgb}{0.95,0.95,0.95}
\definecolor{gray96}{rgb}{0.96,0.96,0.96}
\definecolor{gray97}{rgb}{0.97,0.97,0.97}
\definecolor{gray98}{rgb}{0.98,0.98,0.98}
\definecolor{gray99}{rgb}{0.99,0.99,0.99}
\definecolor{gray9}{rgb}{0.09,0.09,0.09}
\definecolor{gray}{rgb}{0.75,0.75,0.75}
\definecolor{green1}{rgb}{0.00,1.00,0.00}
\definecolor{green2}{rgb}{0.00,0.93,0.00}
\definecolor{green3}{rgb}{0.00,0.80,0.00}
\definecolor{green4}{rgb}{0.00,0.55,0.00}
\definecolor{greenyellow}{rgb}{0.68,1.00,0.18}
\definecolor{green}{rgb}{0.00,1.00,0.00}
\definecolor{grey0}{rgb}{0.00,0.00,0.00}
\definecolor{grey100}{rgb}{1.00,1.00,1.00}
\definecolor{grey10}{rgb}{0.10,0.10,0.10}
\definecolor{grey11}{rgb}{0.11,0.11,0.11}
\definecolor{grey12}{rgb}{0.12,0.12,0.12}
\definecolor{grey13}{rgb}{0.13,0.13,0.13}
\definecolor{grey14}{rgb}{0.14,0.14,0.14}
\definecolor{grey15}{rgb}{0.15,0.15,0.15}
\definecolor{grey16}{rgb}{0.16,0.16,0.16}
\definecolor{grey17}{rgb}{0.17,0.17,0.17}
\definecolor{grey18}{rgb}{0.18,0.18,0.18}
\definecolor{grey19}{rgb}{0.19,0.19,0.19}
\definecolor{grey1}{rgb}{0.01,0.01,0.01}
\definecolor{grey20}{rgb}{0.20,0.20,0.20}
\definecolor{grey21}{rgb}{0.21,0.21,0.21}
\definecolor{grey22}{rgb}{0.22,0.22,0.22}
\definecolor{grey23}{rgb}{0.23,0.23,0.23}
\definecolor{grey24}{rgb}{0.24,0.24,0.24}
\definecolor{grey25}{rgb}{0.25,0.25,0.25}
\definecolor{grey26}{rgb}{0.26,0.26,0.26}
\definecolor{grey27}{rgb}{0.27,0.27,0.27}
\definecolor{grey28}{rgb}{0.28,0.28,0.28}
\definecolor{grey29}{rgb}{0.29,0.29,0.29}
\definecolor{grey2}{rgb}{0.02,0.02,0.02}
\definecolor{grey30}{rgb}{0.30,0.30,0.30}
\definecolor{grey31}{rgb}{0.31,0.31,0.31}
\definecolor{grey32}{rgb}{0.32,0.32,0.32}
\definecolor{grey33}{rgb}{0.33,0.33,0.33}
\definecolor{grey34}{rgb}{0.34,0.34,0.34}
\definecolor{grey35}{rgb}{0.35,0.35,0.35}
\definecolor{grey36}{rgb}{0.36,0.36,0.36}
\definecolor{grey37}{rgb}{0.37,0.37,0.37}
\definecolor{grey38}{rgb}{0.38,0.38,0.38}
\definecolor{grey39}{rgb}{0.39,0.39,0.39}
\definecolor{grey3}{rgb}{0.03,0.03,0.03}
\definecolor{grey40}{rgb}{0.40,0.40,0.40}
\definecolor{grey41}{rgb}{0.41,0.41,0.41}
\definecolor{grey42}{rgb}{0.42,0.42,0.42}
\definecolor{grey43}{rgb}{0.43,0.43,0.43}
\definecolor{grey44}{rgb}{0.44,0.44,0.44}
\definecolor{grey45}{rgb}{0.45,0.45,0.45}
\definecolor{grey46}{rgb}{0.46,0.46,0.46}
\definecolor{grey47}{rgb}{0.47,0.47,0.47}
\definecolor{grey48}{rgb}{0.48,0.48,0.48}
\definecolor{grey49}{rgb}{0.49,0.49,0.49}
\definecolor{grey4}{rgb}{0.04,0.04,0.04}
\definecolor{grey50}{rgb}{0.50,0.50,0.50}
\definecolor{grey51}{rgb}{0.51,0.51,0.51}
\definecolor{grey52}{rgb}{0.52,0.52,0.52}
\definecolor{grey53}{rgb}{0.53,0.53,0.53}
\definecolor{grey54}{rgb}{0.54,0.54,0.54}
\definecolor{grey55}{rgb}{0.55,0.55,0.55}
\definecolor{grey56}{rgb}{0.56,0.56,0.56}
\definecolor{grey57}{rgb}{0.57,0.57,0.57}
\definecolor{grey58}{rgb}{0.58,0.58,0.58}
\definecolor{grey59}{rgb}{0.59,0.59,0.59}
\definecolor{grey5}{rgb}{0.05,0.05,0.05}
\definecolor{grey60}{rgb}{0.60,0.60,0.60}
\definecolor{grey61}{rgb}{0.61,0.61,0.61}
\definecolor{grey62}{rgb}{0.62,0.62,0.62}
\definecolor{grey63}{rgb}{0.63,0.63,0.63}
\definecolor{grey64}{rgb}{0.64,0.64,0.64}
\definecolor{grey65}{rgb}{0.65,0.65,0.65}
\definecolor{grey66}{rgb}{0.66,0.66,0.66}
\definecolor{grey67}{rgb}{0.67,0.67,0.67}
\definecolor{grey68}{rgb}{0.68,0.68,0.68}
\definecolor{grey69}{rgb}{0.69,0.69,0.69}
\definecolor{grey6}{rgb}{0.06,0.06,0.06}
\definecolor{grey70}{rgb}{0.70,0.70,0.70}
\definecolor{grey71}{rgb}{0.71,0.71,0.71}
\definecolor{grey72}{rgb}{0.72,0.72,0.72}
\definecolor{grey73}{rgb}{0.73,0.73,0.73}
\definecolor{grey74}{rgb}{0.74,0.74,0.74}
\definecolor{grey75}{rgb}{0.75,0.75,0.75}
\definecolor{grey76}{rgb}{0.76,0.76,0.76}
\definecolor{grey77}{rgb}{0.77,0.77,0.77}
\definecolor{grey78}{rgb}{0.78,0.78,0.78}
\definecolor{grey79}{rgb}{0.79,0.79,0.79}
\definecolor{grey7}{rgb}{0.07,0.07,0.07}
\definecolor{grey80}{rgb}{0.80,0.80,0.80}
\definecolor{grey81}{rgb}{0.81,0.81,0.81}
\definecolor{grey82}{rgb}{0.82,0.82,0.82}
\definecolor{grey83}{rgb}{0.83,0.83,0.83}
\definecolor{grey84}{rgb}{0.84,0.84,0.84}
\definecolor{grey85}{rgb}{0.85,0.85,0.85}
\definecolor{grey86}{rgb}{0.86,0.86,0.86}
\definecolor{grey87}{rgb}{0.87,0.87,0.87}
\definecolor{grey88}{rgb}{0.88,0.88,0.88}
\definecolor{grey89}{rgb}{0.89,0.89,0.89}
\definecolor{grey8}{rgb}{0.08,0.08,0.08}
\definecolor{grey90}{rgb}{0.90,0.90,0.90}
\definecolor{grey91}{rgb}{0.91,0.91,0.91}
\definecolor{grey92}{rgb}{0.92,0.92,0.92}
\definecolor{grey93}{rgb}{0.93,0.93,0.93}
\definecolor{grey94}{rgb}{0.94,0.94,0.94}
\definecolor{grey95}{rgb}{0.95,0.95,0.95}
\definecolor{grey96}{rgb}{0.96,0.96,0.96}
\definecolor{grey97}{rgb}{0.97,0.97,0.97}
\definecolor{grey98}{rgb}{0.98,0.98,0.98}
\definecolor{grey99}{rgb}{0.99,0.99,0.99}
\definecolor{grey9}{rgb}{0.09,0.09,0.09}
\definecolor{grey}{rgb}{0.75,0.75,0.75}
\definecolor{honeydew1}{rgb}{0.94,1.00,0.94}
\definecolor{honeydew2}{rgb}{0.88,0.93,0.88}
\definecolor{honeydew3}{rgb}{0.76,0.80,0.76}
\definecolor{honeydew4}{rgb}{0.51,0.55,0.51}
\definecolor{honeydew}{rgb}{0.94,1.00,0.94}
\definecolor{hotpink}{rgb}{1.00,0.41,0.71}
\definecolor{indianred}{rgb}{0.80,0.36,0.36}
\definecolor{ivory1}{rgb}{1.00,1.00,0.94}
\definecolor{ivory2}{rgb}{0.93,0.93,0.88}
\definecolor{ivory3}{rgb}{0.80,0.80,0.76}
\definecolor{ivory4}{rgb}{0.55,0.55,0.51}
\definecolor{ivory}{rgb}{1.00,1.00,0.94}
\definecolor{khaki1}{rgb}{1.00,0.96,0.56}
\definecolor{khaki2}{rgb}{0.93,0.90,0.52}
\definecolor{khaki3}{rgb}{0.80,0.78,0.45}
\definecolor{khaki4}{rgb}{0.55,0.53,0.31}
\definecolor{khaki}{rgb}{0.94,0.90,0.55}
\definecolor{lavenderblush}{rgb}{1.00,0.94,0.96}
\definecolor{lavender}{rgb}{0.90,0.90,0.98}
\definecolor{lawngreen}{rgb}{0.49,0.99,0.00}
\definecolor{lemonchiffon}{rgb}{1.00,0.98,0.80}
\definecolor{lightblue}{rgb}{0.68,0.85,0.90}
\definecolor{lightcoral}{rgb}{0.94,0.50,0.50}
\definecolor{lightcyan}{rgb}{0.88,1.00,1.00}
\definecolor{lightgoldenrod}{rgb}{0.93,0.87,0.51}
\definecolor{lightgoldenrod}{rgb}{0.98,0.98,0.82}
\definecolor{lightgray}{rgb}{0.83,0.83,0.83}
\definecolor{lightgreen}{rgb}{0.56,0.93,0.56}
\definecolor{lightgrey}{rgb}{0.83,0.83,0.83}
\definecolor{lightpink}{rgb}{1.00,0.71,0.76}
\definecolor{lightsalmon}{rgb}{1.00,0.63,0.48}
\definecolor{lightsea}{rgb}{0.13,0.70,0.67}
\definecolor{lightsky}{rgb}{0.53,0.81,0.98}
\definecolor{lightslate}{rgb}{0.47,0.53,0.60}
\definecolor{lightslate}{rgb}{0.47,0.53,0.60}
\definecolor{lightslate}{rgb}{0.52,0.44,1.00}
\definecolor{lightsteel}{rgb}{0.69,0.77,0.87}
\definecolor{lightyellow}{rgb}{1.00,1.00,0.88}
\definecolor{limegreen}{rgb}{0.20,0.80,0.20}
\definecolor{linen}{rgb}{0.98,0.94,0.90}
\definecolor{magenta1}{rgb}{1.00,0.00,1.00}
\definecolor{magenta2}{rgb}{0.93,0.00,0.93}
\definecolor{magenta3}{rgb}{0.80,0.00,0.80}
\definecolor{magenta4}{rgb}{0.55,0.00,0.55}
\definecolor{magenta}{rgb}{1.00,0.00,1.00}
\definecolor{maroon1}{rgb}{1.00,0.20,0.70}
\definecolor{maroon2}{rgb}{0.93,0.19,0.65}
\definecolor{maroon3}{rgb}{0.80,0.16,0.56}
\definecolor{maroon4}{rgb}{0.55,0.11,0.38}
\definecolor{maroon}{rgb}{0.69,0.19,0.38}
\definecolor{mediumaquamarine}{rgb}{0.40,0.80,0.67}
\definecolor{mediumblue}{rgb}{0.00,0.00,0.80}
\definecolor{mediumorchid}{rgb}{0.73,0.33,0.83}
\definecolor{mediumpurple}{rgb}{0.58,0.44,0.86}
\definecolor{mediumsea}{rgb}{0.24,0.70,0.44}
\definecolor{mediumslate}{rgb}{0.48,0.41,0.93}
\definecolor{mediumspring}{rgb}{0.00,0.98,0.60}
\definecolor{mediumturquoise}{rgb}{0.28,0.82,0.80}
\definecolor{mediumviolet}{rgb}{0.78,0.08,0.52}
\definecolor{midnightblue}{rgb}{0.10,0.10,0.44}
\definecolor{mintcream}{rgb}{0.96,1.00,0.98}
\definecolor{mistyrose}{rgb}{1.00,0.89,0.88}
\definecolor{moccasin}{rgb}{1.00,0.89,0.71}
\definecolor{navajowhite}{rgb}{1.00,0.87,0.68}
\definecolor{navyblue}{rgb}{0.00,0.00,0.50}
\definecolor{navy}{rgb}{0.00,0.00,0.50}
\definecolor{oldlace}{rgb}{0.99,0.96,0.90}
\definecolor{olivedrab}{rgb}{0.42,0.56,0.14}
\definecolor{orange1}{rgb}{1.00,0.65,0.00}
\definecolor{orange2}{rgb}{0.93,0.60,0.00}
\definecolor{orange3}{rgb}{0.80,0.52,0.00}
\definecolor{orange4}{rgb}{0.55,0.35,0.00}
\definecolor{orangered}{rgb}{1.00,0.27,0.00}
\definecolor{orange}{rgb}{1.00,0.65,0.00}
\definecolor{orchid1}{rgb}{1.00,0.51,0.98}
\definecolor{orchid2}{rgb}{0.93,0.48,0.91}
\definecolor{orchid3}{rgb}{0.80,0.41,0.79}
\definecolor{orchid4}{rgb}{0.55,0.28,0.54}
\definecolor{orchid}{rgb}{0.85,0.44,0.84}
\definecolor{palegoldenrod}{rgb}{0.93,0.91,0.67}
\definecolor{palegreen}{rgb}{0.60,0.98,0.60}
\definecolor{paleturquoise}{rgb}{0.69,0.93,0.93}
\definecolor{paleviolet}{rgb}{0.86,0.44,0.58}
\definecolor{papayawhip}{rgb}{1.00,0.94,0.84}
\definecolor{peachpuff}{rgb}{1.00,0.85,0.73}
\definecolor{peru}{rgb}{0.80,0.52,0.25}
\definecolor{pink1}{rgb}{1.00,0.71,0.77}
\definecolor{pink2}{rgb}{0.93,0.66,0.72}
\definecolor{pink3}{rgb}{0.80,0.57,0.62}
\definecolor{pink4}{rgb}{0.55,0.39,0.42}
\definecolor{pink}{rgb}{1.00,0.75,0.80}
\definecolor{plum1}{rgb}{1.00,0.73,1.00}
\definecolor{plum2}{rgb}{0.93,0.68,0.93}
\definecolor{plum3}{rgb}{0.80,0.59,0.80}
\definecolor{plum4}{rgb}{0.55,0.40,0.55}
\definecolor{plum}{rgb}{0.87,0.63,0.87}
\definecolor{powderblue}{rgb}{0.69,0.88,0.90}
\definecolor{purple1}{rgb}{0.61,0.19,1.00}
\definecolor{purple2}{rgb}{0.57,0.17,0.93}
\definecolor{purple3}{rgb}{0.49,0.15,0.80}
\definecolor{purple4}{rgb}{0.33,0.10,0.55}
\definecolor{purple}{rgb}{0.63,0.13,0.94}
\definecolor{red1}{rgb}{1.00,0.00,0.00}
\definecolor{red2}{rgb}{0.93,0.00,0.00}
\definecolor{red3}{rgb}{0.80,0.00,0.00}
\definecolor{red4}{rgb}{0.55,0.00,0.00}
\definecolor{red}{rgb}{1.00,0.00,0.00}
\definecolor{rosybrown}{rgb}{0.74,0.56,0.56}
\definecolor{royalblue}{rgb}{0.25,0.41,0.88}
\definecolor{saddlebrown}{rgb}{0.55,0.27,0.07}
\definecolor{salmon1}{rgb}{1.00,0.55,0.41}
\definecolor{salmon2}{rgb}{0.93,0.51,0.38}
\definecolor{salmon3}{rgb}{0.80,0.44,0.33}
\definecolor{salmon4}{rgb}{0.55,0.30,0.22}
\definecolor{salmon}{rgb}{0.98,0.50,0.45}
\definecolor{sandybrown}{rgb}{0.96,0.64,0.38}
\definecolor{seagreen}{rgb}{0.18,0.55,0.34}
\definecolor{seashell1}{rgb}{1.00,0.96,0.93}
\definecolor{seashell2}{rgb}{0.93,0.90,0.87}
\definecolor{seashell3}{rgb}{0.80,0.77,0.75}
\definecolor{seashell4}{rgb}{0.55,0.53,0.51}
\definecolor{seashell}{rgb}{1.00,0.96,0.93}
\definecolor{sienna1}{rgb}{1.00,0.51,0.28}
\definecolor{sienna2}{rgb}{0.93,0.47,0.26}
\definecolor{sienna3}{rgb}{0.80,0.41,0.22}
\definecolor{sienna4}{rgb}{0.55,0.28,0.15}
\definecolor{sienna}{rgb}{0.63,0.32,0.18}
\definecolor{skyblue}{rgb}{0.53,0.81,0.92}
\definecolor{slateblue}{rgb}{0.42,0.35,0.80}
\definecolor{slategray}{rgb}{0.44,0.50,0.56}
\definecolor{slategrey}{rgb}{0.44,0.50,0.56}
\definecolor{snow1}{rgb}{1.00,0.98,0.98}
\definecolor{snow2}{rgb}{0.93,0.91,0.91}
\definecolor{snow3}{rgb}{0.80,0.79,0.79}
\definecolor{snow4}{rgb}{0.55,0.54,0.54}
\definecolor{snow}{rgb}{1.00,0.98,0.98}
\definecolor{springgreen}{rgb}{0.00,1.00,0.50}
\definecolor{steelblue}{rgb}{0.27,0.51,0.71}
\definecolor{tan1}{rgb}{1.00,0.65,0.31}
\definecolor{tan2}{rgb}{0.93,0.60,0.29}
\definecolor{tan3}{rgb}{0.80,0.52,0.25}
\definecolor{tan4}{rgb}{0.55,0.35,0.17}
\definecolor{tan}{rgb}{0.82,0.71,0.55}
\definecolor{thistle1}{rgb}{1.00,0.88,1.00}
\definecolor{thistle2}{rgb}{0.93,0.82,0.93}
\definecolor{thistle3}{rgb}{0.80,0.71,0.80}
\definecolor{thistle4}{rgb}{0.55,0.48,0.55}
\definecolor{thistle}{rgb}{0.85,0.75,0.85}
\definecolor{tomato1}{rgb}{1.00,0.39,0.28}
\definecolor{tomato2}{rgb}{0.93,0.36,0.26}
\definecolor{tomato3}{rgb}{0.80,0.31,0.22}
\definecolor{tomato4}{rgb}{0.55,0.21,0.15}
\definecolor{tomato}{rgb}{1.00,0.39,0.28}
\definecolor{turquoise1}{rgb}{0.00,0.96,1.00}
\definecolor{turquoise2}{rgb}{0.00,0.90,0.93}
\definecolor{turquoise3}{rgb}{0.00,0.77,0.80}
\definecolor{turquoise4}{rgb}{0.00,0.53,0.55}
\definecolor{turquoise}{rgb}{0.25,0.88,0.82}
\definecolor{violetred}{rgb}{0.82,0.13,0.56}
\definecolor{violet}{rgb}{0.93,0.51,0.93}
\definecolor{wheat1}{rgb}{1.00,0.91,0.73}
\definecolor{wheat2}{rgb}{0.93,0.85,0.68}
\definecolor{wheat3}{rgb}{0.80,0.73,0.59}
\definecolor{wheat4}{rgb}{0.55,0.49,0.40}
\definecolor{wheat}{rgb}{0.96,0.87,0.70}
\definecolor{whitesmoke}{rgb}{0.96,0.96,0.96}
\definecolor{white}{rgb}{1.00,1.00,1.00}
\definecolor{yellow1}{rgb}{1.00,1.00,0.00}
\definecolor{yellow2}{rgb}{0.93,0.93,0.00}
\definecolor{yellow3}{rgb}{0.80,0.80,0.00}
\definecolor{yellow4}{rgb}{0.55,0.55,0.00}
\definecolor{yellowgreen}{rgb}{0.60,0.80,0.20}
\definecolor{yellow}{rgb}{1.00,1.00,0.00}

\usepackage{amssymb}

\usepackage{epsfig}

\newtheorem{thm}{Theorem}
\newtheorem{lem}{Lemma}
\newtheorem{definition}{Definition}

\newcommand{\DENSE}{%
\setlength{\labelwidth}{12pt}%
\setlength{\labelsep}{2pt}%
\setlength{\leftmargin}{\labelwidth}%
\addtolength{\leftmargin}{\labelsep}%
\setlength{\parsep}{0pt}%
\setlength{\itemsep}{0pt}%
\setlength{\topsep}{0pt}
}

\newenvironment{ditemize}{\begin{list}%
{$-$~}{\DENSE}}{\end{list}}

\font\testo = cmss10 at 12pt
\font\scr = cmss10
\font\scrscr = cmss10 at 6pt
\def\erre{\mathchoice{\hbox{\testo
I\kern-.17emR}}{
   \hbox{\testo I\kern-.17emR}}{\hbox{\scr
I\kern-.17emR}}{
   \hbox{\scrscr I\kern-.17emR}}}
\def\enne{\mathchoice{\hbox{\testo
I\kern-.17emN}}{
   \hbox{\testo I\kern-.17emN}}{\hbox{\scr
I\kern-.17emN}}{
   \hbox{\scrscr I\kern-.17emN}}}

\begin{document}


\Large{\bf Some classes of graphs that are not PCGs} \footnote{Partially supported by {\em Sapienza} University of Rome, project ``Combinatorial structures and algorithms for problems in co-phylogeny''.\\
Part of the results of this paper have been submitted to a conference.}

\vspace*{1cm}
\begin{centering}
\large{Pierluigi Baiocchi \,\,\,\,\,\, Tiziana Calamoneri} \\
\large{Angelo Monti \,\,\,\,\,\,\,\,\,Rossella Petreschi}

\vspace*{.5cm}

\normalsize

Computer Science Department,\\ 
``Sapienza'' University of Rome, Italy\\ 
pierluigi.baiocchi@gmail.com, \{calamo,monti,petreschi\}@di.uniroma1.it\\
\,\,
\end{centering}

\begin{abstract}
A graph $G=(V,E)$ is a {\em pairwise compatibility graph} (PCG) if there 
exists an edge-weighted tree $T$ and  two 
non-negative real numbers $d_{min}$ and $d_{max}$,  
$d_{min} \leq d_{max}$, such that each node $u \in V$ is uniquely 
associated to a leaf of $T$ and there is an edge $(u,v) \in E$ 
if and only if $d_{min} \leq d_{T} (u, v) \leq d_{max}$,
where $d_{T} (u, v)$ is the sum of the weights of the edges on the unique 
path $P_{T}(u,v)$ from $u$ to $v$ in $T$. 
Understanding which graph classes lie inside and which ones outside the PCG class is an important issue.
In this paper we propose a new proof technique that allows us to show that some interesting classes of graphs have empty intersection with PCG. 
As an example, we use this technique to show that wheels and graphs obtained as strong product between a cycle and $P_2$ are not PCGs.
\end{abstract}

\normalsize

{\bf keywords:}
Phylogenetic Tree Reconstruction Problem,
Pairwise Compatibility Graphs (PCGs),
PCG Recognition Problem,
Wheel.

%

\section{Introduction}

Graphs we deal with in this paper are motivated by a fundamental problem in computational biology, that is
the reconstruction of ancestral relationships \cite{O99}. 
It is known that the evolutionary history of a set of organisms is represented by a {\em phylogenetic tree}, i.e. a tree where leaves represent distinct known taxa while the internal nodes possible ancestors that might have led through evolution to this set of taxa. The edges of the tree are weighted in order to represent a kind of evolutionary distance among species.  
Given a set of taxa, the {\em phylogenetic tree reconstruction problem} consists in finding the ``best'' phylogenetic tree that explains the given data.  
Since it is not completely clear what ``best'' means, the performance of the reconstruction algorithms is usually evaluated experimentally by comparing the tree produced by the algorithm with those partial subtrees that are unanimously recognized as ``sure'' by biologists. 
However, the tree reconstruction problem is proved to be NP-hard under many criteria of optimality, moreover real phylogenetic trees  are usually huge, so testing these heuristics on real data is in general very difficult. This is the reason why it is common to exploit {\em sample techniques}, extracting relatively small subsets of taxa from large 
phylogenetic trees, according to some biologically-motivated constraints, and to test the reconstruction algorithms only on the smaller subtrees induced by the sample. 
The underlying idea is that the behavior of the algorithm on the whole tree will be more or less the same as on the sample. 
It has been observed that using in the sample very close or very distant taxa can create problems for phylogeny reconstruction algorithms \cite{F78} so, in selecting a sample from the leaves of  the tree, the constraint of keeping the pairwise distance between any two leaves in the sample between two given positive integers $d_{min}$ and $d_{max}$ is used.
This motivates the introduction of \emph{pairwise compatibility graphs} (PCGs): given  a phylogenetic tree $T$, and integers $d_{min}, d_{max}$ we can associate a graph $G$, called the pairwise compatibility graph of $T$,  whose nodes are the leaves  of $T$ and for which there is an edge between two nodes if the corresponding leaves in $T$ are at weighted distance within the interval $[d_{min}, d_{max}]$. 

From a more theoretical point of view, we highlight that the problem of sampling a set of $m$ leaves from a weighted tree $T$, such that their pairwise distance is within some interval $[d_{min},d_{max}]$, reduces to selecting a clique of size $m$ uniformly  at random from the associated pairwise compatibility graph.
As the sampling problem can be solved in polynomial time on PCGs \cite{KMP03}, it follows that the {\em max clique problem} is solved in polynomial time on this class of graphs, if the edge-weighted tree $T$ and the two values $d_{min},d_{max}$ are  known or can be provided in polynomial time.

\medskip

The previous reasonings motivate the interest of researchers in the so called {\em PCG recognition problem}, consisting in understanding whether, given a graph $G$, it is possible to determine an edge-weighted tree $T$ and two integers $d_{min}, d_{max}$ such that $G$ is the associated pairwise compatibility graph.

\begin{figure}[h]
\begin{center}
\begin{tabular}{c  c  c}
\begin{picture}(70,90)(0,0)
\put(10,50){\circle*{6}}
\put(30,30){\circle*{6}}
\put(30,70){\circle*{6}}
\put(50,50){\circle*{6}}
\put(6,40){$a$}
\put(26,18){$d$}
\put(26,75){$b$}
\put(46,40){$c$}
\qbezier(10,50)(30,70)(30,70)
\qbezier(30,30)(50,50)(50,50)
\put(30,70){\line(0,-1){40}}
\end{picture}
&
%
%
\begin{picture}(110,90)(0,0)
\put(0,30){\circle*{6}}
\put(30,30){\circle*{6}}
\put(60,30){\circle*{6}}
\put(90,30){\circle*{6}}
\put(-4,20){$a$}
\put(26,20){$d$}
\put(56,20){$b$}
\put(86,20){$c$}

\put(30,60){\circle*{6}}
\put(60,60){\circle*{6}}

\put(30,60){\line(0,-1){30}}
\put(60,60){\line(0,-1){30}}

\put(30,60){\line(1,0){30}}

\qbezier(0,30)(30,60)(30,60)
\qbezier(90,30)(60,60)(60,60)

\put(3,40){2}
\put(30,40){1}
\put(60,40){1}
\put(83,40){2}

\put(42,63){2}
\end{picture}
&
\begin{picture}(80,90)(0,0)
\color{red}
\thicklines

\dottedline{4}(36,70)(56,50)
\dottedline{4}(14,50)(36,30)
\color{blue}
\dashline{7}(14,50)(56,50)
\color{black}
\put(14,50){\circle*{6}}
\put(36,30){\circle*{6}}
\put(36,70){\circle*{6}}
\put(56,50){\circle*{6}}
\put(8,40){$a$}
\put(32,18){$d$}
\put(32,75){$b$}
\put(52,40){$c$}
\qbezier(14,50)(36,70)(36,70)
\qbezier(36,30)(56,50)(56,50)
\put(36,70){\line(0,-1){40}}
\end{picture}
\\
a. & \hspace*{-.5cm}b. & \hspace*{-.5cm}c.
\end{tabular}
\end{center}
\vspace*{-0.5cm}
\caption{a. A graph $G$. 
b. An edge-weighted caterpillar $T$ such that $G=PCG(T, 4, 5)$. c. $G$ where the PCG-coloring induced by triple $T, 4, 5$ is highlighted.}
\label{fig:PCG}
\end{figure}
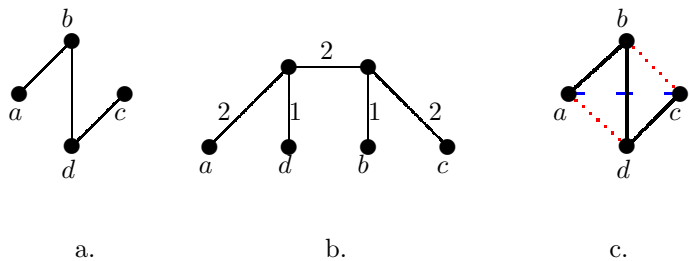

In Figure \ref{fig:PCG}.a a small graph that is  $PCG(T, 4,5)$ is depicted and, in Figure \ref{fig:PCG}.b, $T$ is shown.
 In general, $T$ is not unique; here $T$ is a {\em caterpillar}, i.e. a tree consisting of a central path, called spine, and nodes directly connected to that path.
Due to their simple structure, caterpillars are the most used witness trees to show that a graph is PCG.
However, it has been proven that there are some PCGs for which it is not possible to find a caterpillar as witness tree \cite{CFS13}.

Due to the flexibility afforded in the construction of instances (i.e. choice of tree topology and values for $d_{min}$ and $d_{max}$), when PCGs were introduced, it was also conjectured that all graphs are PCGs \cite{KMP03}.
This conjecture has been confuted by proving the existence of some graphs not belonging to PCG.
Namely,
Yanhaona et al. \cite{Yal10} showed a not PCG bipartite graph with 15 nodes (Figure~\ref{fig.notPCG}.a). 
Subsequently, Mehnaz and Rahman \cite{MR13} generalized the used technique to provide a class of bipartite graphs that are not PCGs.
More recently, Durochet et al. \cite{DMR15} proved that there exists a not bipartite graph with 8 nodes that is not PCG (Figure~\ref{fig.notPCG}.b); this is the smallest graph that is not PCG, since all graphs with at most 7 nodes are PCGs \cite{CFS13}.

The authors of \cite{DMR15} provided also an example of a planar graph with 20 nodes that is not PCG (Figure~\ref{fig.notPCG}.c). 
Finally, it holds that, if a graph $H$ is not PCG, every graph admitting $H$ as induced subgraph is not PCG, too \cite{survey}. 

On the other side, many graph classes have been proved to be in PCG, such as cliques and trees, cycles, single chord cycles, cacti, tree power graphs 
\cite{Yal09,Yal10}, interval graphs \cite{BH08} Dilworth 2 and dilworth $k$ graphs \cite{CP14, CP14k}.
\begin{figure}[t]
\begin{center}
\vspace*{1cm}
\hspace*{-2.7cm}
\begin{tabular}{c  c  c}
\begin{picture}(120,80)(0,0)
\put(25,10){\circle*{6}}
\put(45,10){\circle*{6}}
\put(65,10){\circle*{6}}
\put(85,10){\circle*{6}}
\put(105,10){\circle*{6}}

\put(0,60){\circle*{6}}
\put(15,60){\circle*{6}}
\put(30,60){\circle*{6}}
\put(45,60){\circle*{6}}
\put(60,60){\circle*{6}}
\put(75,60){\circle*{6}}
\put(90,60){\circle*{6}}
\put(105,60){\circle*{6}}
\put(120,60){\circle*{6}}
\put(135,60){\circle*{6}}

\qbezier(25,10)(0,60)(0,60)
\qbezier(25,10)(15,60)(15,60)
\qbezier(25,10)(30,60)(30,60)
\qbezier(25,10)(45,60)(45,60)
\qbezier(25,10)(60,60)(60,60)
\qbezier(25,10)(75,60)(75,60)

\qbezier(45,10)(0,60)(0,60)
\qbezier(45,10)(15,60)(15,60)
\qbezier(45,10)(30,60)(30,60)
\qbezier(45,10)(90,60)(90,60)
\qbezier(45,10)(105,60)(105,60)
\qbezier(45,10)(120,60)(120,60)

\qbezier(65,10)(0,60)(0,60)
\qbezier(65,10)(45,60)(45,60)
\qbezier(65,10)(60,60)(60,60)
\qbezier(65,10)(90,60)(90,60)
\qbezier(65,10)(105,60)(105,60)
\qbezier(65,10)(135,60)(135,60)

\qbezier(85,10)(15,60)(15,60)
\qbezier(85,10)(45,60)(45,60)
\qbezier(85,10)(75,60)(75,60)
\qbezier(85,10)(105,60)(105,60)
\qbezier(85,10)(120,60)(120,60)
\qbezier(85,10)(135,60)(135,60)

\qbezier(105,10)(30,60)(30,60)
\qbezier(105,10)(60,60)(60,60)
\qbezier(105,10)(75,60)(75,60)
\qbezier(105,10)(105,60)(105,60)
\qbezier(105,10)(120,60)(120,60)
\qbezier(105,10)(135,60)(135,60)
\end{picture}
&
\begin{picture}(100,80)(0,0)
\put(35,0){\circle*{6}}
\put(35,20){\circle*{6}}
\put(0,35){\circle*{6}}
\put(20,35){\circle*{6}}
\put(50,35){\circle*{6}}
\put(70,35){\circle*{6}}
\put(35,50){\circle*{6}}
\put(35,70){\circle*{6}}


\put(0,35){\line(1,0){20}}
\put(50,35){\line(1,0){20}}
\put(35,0){\line(0,1){20}}
\put(35,50){\line(0,1){20}}

\qbezier(0,35)(35,20)(35,20)
\qbezier(0,35)(35,0)(35,0)
\qbezier(0,35)(35,50)(35,50)
\qbezier(0,35)(35,70)(35,70)
\qbezier(20,35)(35,20)(35,20)
\qbezier(20,35)(35,0)(35,0)
\qbezier(20,35)(35,50)(35,50)
\qbezier(20,35)(35,70)(35,70)

\qbezier(50,35)(35,20)(35,20)
\qbezier(50,35)(35,0)(35,0)
\qbezier(50,35)(35,50)(35,50)
\qbezier(50,35)(35,70)(35,70)
\qbezier(70,35)(35,20)(35,20)
\qbezier(70,35)(35,0)(35,0)
\qbezier(70,35)(35,50)(35,50)
\qbezier(70,35)(35,70)(35,70)
\end{picture}
&
\hspace*{-1.8cm}
\begin{picture}(100,80)(0,0)
\put(10,50){\circle*{6}}
\put(20,50){\circle*{6}}
\put(30,50){\circle*{6}}
\put(40,50){\circle*{6}}

\put(60,50){\circle*{6}}
\put(70,50){\circle*{6}}
\put(80,50){\circle*{6}}
\put(90,50){\circle*{6}}

\put(110,50){\circle*{6}}
\put(120,50){\circle*{6}}
\put(130,50){\circle*{6}}
\put(140,50){\circle*{6}}

\put(160,50){\circle*{6}}
\put(170,50){\circle*{6}}
\put(180,50){\circle*{6}}
\put(190,50){\circle*{6}}

\put(90,20){\circle*{6}}
\put(90,80){\circle*{6}}
\put(160,20){\circle*{6}}
\put(160,80){\circle*{6}}

\put(10,50){\line(1,0){10}}
\put(30,50){\line(1,0){10}}
\put(60,50){\line(1,0){10}}
\put(80,50){\line(1,0){10}}
\put(110,50){\line(1,0){10}}
\put(130,50){\line(1,0){10}}
\put(160,50){\line(1,0){10}}
\put(180,50){\line(1,0){10}}

\qbezier(90,80)(60,50)(60,50)
\qbezier(90,80)(70,50)(70,50)
\qbezier(90,80)(80,50)(80,50)
\qbezier(90,80)(90,50)(90,50)
\qbezier(90,20)(60,50)(60,50)
\qbezier(90,20)(70,50)(70,50)
\qbezier(90,20)(80,50)(80,50)
\qbezier(90,20)(90,50)(90,50)

\qbezier(90,80)(110,50)(110,50)
\qbezier(90,80)(120,50)(120,50)
\qbezier(90,80)(130,50)(130,50)
\qbezier(90,80)(140,50)(140,50)
\qbezier(160,20)(110,50)(110,50)
\qbezier(160,20)(120,50)(120,50)
\qbezier(160,20)(130,50)(130,50)
\qbezier(160,20)(140,50)(140,50)

\qbezier(160,80)(160,50)(160,50)
\qbezier(160,80)(170,50)(170,50)
\qbezier(160,80)(180,50)(180,50)
\qbezier(160,80)(190,50)(190,50)
\qbezier(160,20)(160,50)(160,50)
\qbezier(160,20)(170,50)(170,50)
\qbezier(160,20)(180,50)(180,50)
\qbezier(160,20)(190,50)(190,50)

\qbezier(10,50)(50,-20)(90,20)
\qbezier(20,50)(50,-10)(90,20)
\qbezier(30,50)(50,0)(90,20)
\qbezier(40,50)(50,10)(90,20)

\qbezier(10,50)(80,140)(160,80)
\qbezier(20,50)(80,130)(160,80)
\qbezier(30,50)(80,120)(160,80)
\qbezier(40,50)(80,110)(160,80)
\end{picture}
\\
\hspace*{.4cm}a.&
\hspace*{-.9cm}b.&
\hspace*{2cm}c.
\end{tabular}
\end{center}
\caption{a. The first graph proven not to be a PCG. b. The graph of smallest size proven not to be a PCG. c. A planar graph that is not PCG.}
\label{fig.notPCG}
\end{figure}
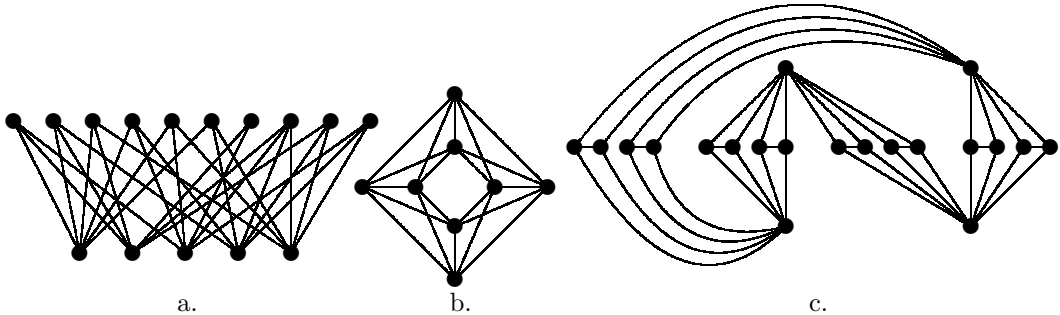

However, despite these results,
it remains unclear which is the boundary of the PCG class.
In this paper, we move a step in the direction of searching new graph classes that are not PCGs.
Indeed, in Section \ref{sec.technique} we introduce a new general proof technique that allows us to show that a graph is not a PCG. We exploit it on two interesting classes of graphs:

\begin{itemize}
\item wheels, for which it was left as an open problem to understand whether they were PCGs or not \cite{CAS13};
\item  graphs obtained as strong product between a cycle and $P_2$, that are a generalization of the smallest known not PCG \cite{DMR15}.
\end{itemize}

After some preliminaries (Section \ref{sec.forbidden}),
the results dealing with these classes are presented in Sections \ref{sec.wheel} and \ref{ref.product}, respectively.

Finally, in Section \ref{sec.minimality}, for any graph $G$ in each one of the two classes, we show that by deleting any node from $G$ we get a PCG, so proving that it does not contain any induced subgraph that is not PCG.

We conclude the paper with Section \ref{sec.conclusion}, where we address some open problems.

\section{Proof Technique}
\label{sec.technique}

In this section, after introducing some definitions, we describe our proof technique, useful to prove that some classes of graphs have empty intersection with the class of PCGs, formally defined as follows.

\begin{definition} \cite{KMP03}
A graph $G=(V,E)$ is a {\em pairwise compatibility graph} (PCG) if there 
exist a tree $T$, a weight function assigning a positive real value to each edge of $G$, and  two 
non-negative real numbers $d_{min}$ and $d_{max}$,  
$d_{min} \leq d_{max}$, such that each node $u \in V$ is uniquely 
associated to a leaf of $T$ and there is an edge $(u,v) \in E$ 
if and only if $d_{min} \leq d_{T} (u, v) \leq d_{max}$,
where $d_{T} (u, v)$ is the sum of the weights of the edges on the unique 
path $P_{T}(u,v)$ from $u$ to $v$ in $T$. 
In such a case, we say that $G$ is a 
PCG of $T$ for $d_{min}$ and $d_{max}$; in symbols, $G=PCG(T, d_{min}, d_{max})$.
\end{definition}

In order not to overburden the exposition, in the following,
when we speak about a tree, we implicitly mean that it is edge-weighted. 

\medskip

Given a graph $G=(V,E)$, 
we call {\em non-edges} of $G$ the edges that do not belong to the graph.
A {\em tri-coloring of $G$} is an edge labeling of the complete graph $K_{|V|}$ with labels from set $\{$ black, red, blue $\}$ such that all edges of $K_{|V|}$ that are in $G$ are labeled black, while the other edges of $K_{|V|}$ (i.e. the non-edges of $G$) are labeled either red or blue.
A tri-coloring is called a {\em partial tri-coloring} if not all the 
non-edges of $G$ are labeled.

Notice that, if $G=PCG(T, d_{min}, d_{max})$, some of its non-edges do not belong to $G$ because the weights of the corresponding paths on $T$ are strictly larger than $d_{max}$, while some other edges are not in $G$ because the weights of the corresponding paths on $T$ are strictly smaller than $d_{min}$. This motivates the following definition.

\begin{definition}
Given a graph $G=PCG(T, d_{min}, d_{max})$, we call its {\em PCG-coloring} the tri-coloring $\mathcal{C}$ of $G$ such that:\\
- $(u,v)$ is {\em red} in $\mathcal{C}$ if $d_T(u,v) < d_{min}$, \\ 
- $(u,v)$ is {\em black} in $\mathcal{C}$ if $d_{min} \leq d_T(u,v) \leq d_{max}$, \\
- $(u,v)$ is {\em blue} in $\mathcal{C}$ if $d_T(u,v) > d_{max}$. 

\noindent
In such a case, we say that triple $(T, d_{min}, d_{max})$ induces PCG-coloring $\mathcal{C}$.
\end{definition}

In order to read the figures even in gray scale, we draw red edges as red and dotted and blue edges as blue and dashed in all the figures.

In Figure  \ref{fig:PCG}.c we highlight the PCG-coloring induced by triple $(T, 4, 5)$
where $T$ is the tree in Figure \ref{fig:PCG}.b.

The following definition formalizes that not all tri-colorings are PCG-colorings.

\begin{definition}
A tri-coloring $\mathcal{C}$ (either partial or not) of a graph $G$
is called a {\em forbidden PCG-coloring} if no triple $(T, d_{min}, d_{max})$ inducing $\mathcal{C}$ exists.
\end{definition}

Observe that a graph is PCG if and only if there exists a tri-coloring $\mathcal{C}$ that is a PCG-coloring for $G$.

Besides, any induced subgraph $H$ of a given $G=PCG(T, d_{min}, d_{max})$ is also PCG, indeed $H=PCG(T', d_{min}, d_{max})$, where $T'$ is the 
subtree induced by the leaves corresponding to the nodes of $H$.
Moreover, $H$ inherits the PCG-coloring induced by triple $(T, d_{min}, d_{max})$ from $G$.
Thus, if we were able to prove that $H$ inherits a forbidden PCG-coloring from a tri-coloring $\mathcal{C}$ of $G$, then we would show that $\mathcal{C}$ cannot be a PCG-coloring for $G$ in any way.
This is the core of our proof technique.

\medskip

\noindent {\bf Tecnique:}\\
{\em
Given a graph $G$ that we want to prove not to be PCG:
\begin{enumerate}
\item
list some forbidden PCG-colorings of particular graphs that are induced pairwise compatibility subgraphs of $G$;

\item
show that each tri-coloring of $G$ induces a forbidden PCG-coloring in at least an induced subgraph;

\item
conclude that $G$ is not PCG, since all its tri-colorings are proved to be forbidden.
\end{enumerate}
}

\section{Forbidden Tri-Colorings}
\label{sec.forbidden}

We now highlight some forbidden partial tri-colorings.
In agreement with the proof technique described in the previous section, 
alongl the paper, we will use them to show that the three considered classes have empty intersection with PCG.
\medskip

Given a graph $G=(V,E)$ and a subset 
$S \subseteq V$, we denote by 
$G[S]$ the subgraph of $G$ induced by nodes in $S$. 

A {\em subtree induced by a set of leaves} of $T$ is the minimal subtree of $T$ which contains those leaves. In particular, we denote by $T_{uvw}$ the subtree of a tree induced by three leaves $u, v$ and $w$.

The following lemma from \cite{Yal10} will be largely used:

\begin{lem} 
\label{lemma.subtree}
Let $T$ be a tree, and $u,v$ and $w$ be three leaves of $T$ such that $P_T(u, v)$ is the largest path in $T_{uvw}$. 
Let $x$ be a leaf of $T$ other than $u, v, w$. Then, $d_T(w,x) \leq \max \{d_T(u,x), d_T(v,x) \}$.
\end{lem}

It is immediate to see that the $m$ node path, $P_m$, is a PCG; the following lemma gives some constraints to the associated PCG-coloring.
\begin{lem}
\label{claim1}
Let $P_m$, $m\geq 4$, be a path and let $\mathcal{C}$ be one of its PCG-colorings.
If all non-edges $(v_1,v_i)$, $3 \leq i \leq m-1$, and $(v_2,v_m)$ are colored with blue in $\mathcal{C}$, then also non-edge $(v_1,v_m)$ is  colored with blue in $\mathcal{C}$.
\end{lem}
\begin{proof}
Let $\mathcal{C}$ be the PCG-coloring of $P_m$ induced by triple $(T, d_{min}, d_{max})$.
We apply Lemma \ref{lemma.subtree} iteratively. 

First consider nodes $v_1$, $v_2$, $v_3$ and $v_4$ as $u$, $w$, $v$ and $x$: $P_T(v_1, v_3)$ is easily the largest path in $T_{v_1 v_3 v_2}$;  then $d_T(v_2,v_4) \leq \max \{d_T(v_1, v_4), d_T(v_3, v_4) \}=d_T(v_1, v_4)$ because $(v_1, v_4)$ is a blue non-edge by hypothesis while $(v_3, v_4)$ is an edge. 

Now repeat the reasoning with nodes $v_1$, $v_2$, $v_i$ and $v_{i+1}$, $4 \leq i < m$, as $u$, $w$, $v$ and $x$, exploiting that at the previous step we have obtained that $d_T(v_2, v_i) \leq d_T(v_1, v_i)$: in $T_{v_1 v_i v_2}$, $P_T(v_1, v_i)$ is the largest path and so $d_T(v_2,v_{i+1}) \leq \max \{ d_T(v_1, v_{i+1}),$ $d_T(v_i, v_{i+1})\}$ $=d_T(v_1, v_{i+1})$ since $(v_1, v_{i+1})$ is a blue non-edge while $(v_i, v_{i+1})$ is an edge.

Posing $i=m-1$, we get that  $d_T(v_2,v_m) \leq d_T(v_1, v_m)$; since non-edge $(v_2, v_m)$ is blue by hypothesis, $(v_1, v_m)$ is blue, too.
\end{proof}

%
%

Given a graph, in order to ease the exposition, we call {\em 2-non-edge} a non-edge between nodes that are at distance $2$ in the graph.

\begin{lem}
\label{le1}
Let $P_n$, $n\geq 3$, be a path. Any $PCG$-coloring of $P_n$ that has at least one red non-edge but no red 2-non-edges is forbidden.
\end{lem}
\begin{proof}
If $n=3$, there is a unique non-edge and it is a 2-non-edge; so, the claim trivially follows.

If $n \geq 4$, consider a triple $(T, d_{min}, d_{max})$ inducing a PCG-coloring with at least a red non-edge.
Among all red non-edges, let $(v_i, v_j)$ be the one such that $j-i$ is minimum. Assume by contradiction, $j-i>2$.
Consider now the subpath $P'$ induced by $v_i, \ldots , v_j$.
$P'$ has at least 4 nodes and inherits the PCG-coloring from $P_n$; in it, there is only a red non-edge (i.e. the non-edge connecting $v_i$ and $v_j$).
$P'$ satisfies the hypothesis of Lemma \ref{claim1}, hence $(v_i, v_j)$ must be blue, against the hypothesis that it is red.
\end{proof}

The following lemma is proved in \cite{Yal09} and here translated in our setting:

\begin{lem}
\label{lemma.cicloRossoBlu}
In every PCG-coloring of the $n$ node cycle $C_n$, $n\geq 4$, there exist at least one red and one blue non-edges.
\end{lem}

\begin{thm}
\label{cicloDistanza2}
Let $C_n$, $n\geq 4$,  be a cycle. Then any $PCG$-coloring of $C_n$ that has no red 2-non-edges  is forbidden.
\end{thm}
\begin{proof}
Let $C_n=PCG(T, d_{min}, d_{max})$, $n \geq 4$; from Lemma \ref{lemma.cicloRossoBlu}, there exists at least a red non-edge.
W.l.o.g. assume that this non-edge is $(v_1,v_i)$, with $4\leq i < n-1$. 
We apply Lemma \ref{le1} on the induced $P_i$   and the thesis follows by contradiction.
\end{proof}

\begin{figure}[t]
\begin{center}
\vspace*{.5cm}
\begin{tabular}{c  c  c  c  c}
\begin{picture}(60,80)(0,0)
\color{red}
\thicklines
\dottedline{4}(10,30)(50,70)
\dottedline{4}(10,70)(50,30)
\dottedline{4}(10,70)(50,70)
\dottedline{4}(10,30)(50,30)

\color{black}
\thicklines
\put(10,30){\circle*{6}}
\put(10,70){\circle*{6}}
\put(50,30){\circle*{6}}
\put(50,70){\circle*{6}}
\put(10,18){$a$}
\put(10,75){$b$}
\put(50,75){$c$}
\put(50,18){$d$}
\put(10,70){\line(0,-1){40}}
\put(50,70){\line(0,-1){40}}
\thinlines
\end{picture}

&
\begin{picture}(60,80)(0,0)
\color{red}
\thicklines
\dottedline{4}(10,70)(50,70)
\dottedline{4}(10,30)(50,30)

\color{blue}
\dashline{7}(10,30)(50,70)
\dashline{7}(10,70)(50,30)

\color{black}
\thicklines
\put(10,30){\circle*{6}}
\put(10,70){\circle*{6}}
\put(50,30){\circle*{6}}
\put(50,70){\circle*{6}}
\put(10,20){$a$}
\put(10,75){$b$}
\put(50,75){$c$}
\put(50,20){$d$}
\put(10,70){\line(0,-1){40}}
\put(50,70){\line(0,-1){40}}
\thinlines
\end{picture}

&
\begin{picture}(60,80)(0,0)
\color{red}
\thicklines
\dottedline{4}(10,30)(50,30)

\color{blue}
\dashline{7}(10,30)(50,70)
\dashline{7}(10,70)(50,30)

\color{black}
\thicklines
\put(10,70){\line(1,0){40}}
\put(10,30){\circle*{6}}
\put(10,70){\circle*{6}}
\put(50,30){\circle*{6}}
\put(50,70){\circle*{6}}
\put(10,18){$a$}
\put(10,75){$b$}
\put(50,75){$c$}
\put(50,18){$d$}
\put(10,70){\line(0,-1){40}}
\put(50,70){\line(0,-1){40}}
\thinlines
\end{picture}
&
\begin{picture}(60,80)(0,0)
\color{red}
\thicklines
\dottedline{4}(30,70)(10,30)
\dottedline{4}(30,70)(50,30)

\color{blue}
\dashline{7}(10,30)(50,30)

\color{black}
\thicklines
\put(10,30){\circle*{6}}
\put(30,70){\circle*{6}}
\put(30,50){\circle*{6}}
\put(50,30){\circle*{6}}
\put(6,18){$a$}
\put(30,38){$b$}
\put(26,75){$d$}
\put(46,18){$c$}
\qbezier(10,30)(30,50)(30,50)
\qbezier(50,30)(30,50)(30,50)
\put(30,70){\line(0,-1){20}}
\thinlines
\end{picture}

&

\begin{picture}(55,60)(0,0)
\color{red}
\thicklines
\dottedline{4}(0,55)(25,30)
\dottedline{4}(0,55)(50,55)

\color{blue}
\dashline{7}(0,55)(25,80)

\color{black}
\thicklines
\put(0,55){\circle*{6}}
\put(50,55){\circle*{6}}
\put(25,30){\circle*{6}}
\put(25,80){\circle*{6}}
\put(-2,45){$a$}
\put(25,85){$b$}
\put(25,20){$d$}
\put(50,45){$c$}
\qbezier(25,30)(50,55)(50,55)
\qbezier(25,80)(50,55)(50,55)
\put(25,30){\line(0,1){50}}
\thinlines
\end{picture}

\\
\hspace*{-.5cm} a. {\bf f-c}$(2K_2)a$&
\hspace*{-.5cm} b. {\bf f-c}$(2K_2)b$&
\hspace*{-.5cm} c. {\bf f-c}$(P_4)$&
\hspace*{-.5cm} d. {\bf f-c}$(K_{1,3})$&
\hspace*{-.5cm} e. {\bf f-c}$(K_3 \cup K_1)$
\end{tabular}
\end{center}
\caption{Some forbidden tri-colorings of small graphs. Acronym {\bf f-c} stands for {\em forbidden coloring}.}
\label{fig:forbidden}
\end{figure}
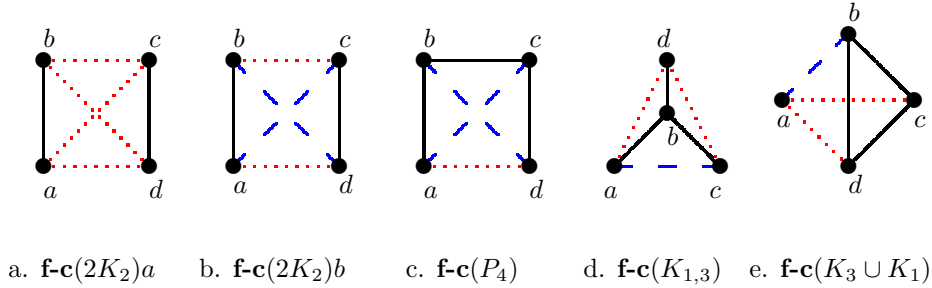

\begin{thm}
\label{lemma.forbidden}
The  tri-colorings in Figure \ref{fig:forbidden} are forbidden PCG-colorings.
\end{thm}

\begin{proof} We prove separately that the tri-colorings in figure are forbidden for PCGs $2K_2$, $P_4$, $K_{1,3}$ and $K_3 \cup K_1$.

\medskip

\bigskip

\bigskip

\noindent
{\bf Forbidden tri-coloring f-c$(2K_2)a$}:

We obtain that the tri-coloring in Figure \ref{fig:forbidden}.a is forbidden by rephrasing Lemma 6 of  \cite{DMR15} with our nomenclature.

\medskip

The other proofs are all by contradiction and proceed as follows: for each tri-coloring in Figure \ref{fig:forbidden}, we assume that it is a feasible PCG-coloring induced by a triple $(T, d_{min},d_{max})$ and show that this assumption contradicts Lemma \ref{lemma.subtree}.

\medskip
\noindent
{\bf Forbidden tri-coloring f-c$(2K_2)b$}:

From the tri-coloring in Figure \ref{fig:forbidden}.b we have that
$$d_T(b,c) < d_{min}\leq d_T(a,b)\leq  d_{max}< d_T(a,c). 
$$ 
Thus $P_T(a,c)$ is the largest path in $T_{a,b,c}$. By Lemma \ref{lemma.subtree}, for leaf $d$ it must be: $d_T(b,d) \leq \max\left\{d_T(a,d), d_T(c,d)\right\}= d_T(c,d)$ while from the tri-coloring it holds that $d_T(c,d) \leq d_{max}<d_T(b,d)$, a contradiction.

\medskip
\noindent
{\bf Forbidden tri-coloring f-c$(P_4)$}:

From the tri-coloring in Figure \ref{fig:forbidden}.c we have that 
$$d_T(a,b), d_T(b,c)  \leq   d_{max}< d_T(a,c).
$$ 
Thus $P_T(a,c)$ is the largest path in $T_{a,b,c}$. By Lemma \ref{lemma.subtree}, for leaf $d$ we have:
$d_T(b,d) \leq \max\left\{d_T(a,d), d_T(c,d)\right\}= d_T(c,d)$ while from the tri-coloring it holds that $d_T(c,d) \leq d_{max}<d_T(b,d)$, a contradiction.

\medskip
\noindent
{\bf Forbidden tri-coloring f-c$(K_{1,3})$}:

From the tri-coloring in Figure \ref{fig:forbidden}.d we have that 
$$d_T(a,b), d_T(b,c)  \leq   d_{max}< d_T(a,c).
$$ 
Thus $P_T(a,c)$ is the largest path in $T_{a,b,c}$. By Lemma \ref{lemma.subtree}, for leaf $d$ we have:
$d_T(b,d) \leq \max\left\{d_T(a,d), d_T(c,d)\right\}$ while from the tri-coloring it holds that $d_T(a,d), d_T(c,d) < d_{min}\leq d_T(b,d)$.

\medskip
\noindent
{\bf Forbidden tri-coloring f-c$(K_3 \cup K_1)$}:

From the tri-coloring in Figure \ref{fig:forbidden}.e we have that
$$d_T(a,d), d_T(a,c)< d_{min}\leq d_T(c,d). 
$$ 
Thus $P_T(c,d)$ is the largest path in $T_{a,c,d}$. By Lemma \ref{lemma.subtree}, for leaf $b$ it must be: $d_T(a,b) \leq \max\left\{d_T(c,b), d_T(d,b)\right\}$ while from the tri-coloring it holds that $d_T(c,b),$ $d_T(d,b) \leq d_{max} < d_T(a,b)$, a contradiction.
\end{proof}

\begin{figure}[t]
\begin{center}
\vspace*{0.5cm}
\hspace*{1cm}
\begin{tabular}{c  c  c }
\begin{picture}(90,80)(0,0)
\color{blue}
\thicklines
\dashline{7}(10,30)(55,55)
\dashline{7}(45,30)(0,55)

\color{black}
\qbezier(0,55)(10,30)(10,30)
\qbezier(0,55)(30,85)(30,85)
\qbezier(0,55)(30,65)(30,65)
\qbezier(55,55)(45,30)(45,30)
\qbezier(55,55)(30,85)(30,85)
\qbezier(55,55)(30,65)(30,65)

\put(30,65){\line(0,1){20}}
\put(10,30){\line(1,0){35}}

\put(0,55){\circle*{6}}
\put(10,30){\circle*{6}}
\put(30,85){\circle*{6}}
\put(30,65){\circle*{6}}
\put(45,30){\circle*{6}}
\put(55,55){\circle*{6}}
\put(2,60){$e$}
\put(4,18){$d$}
\put(52,60){$b$}
\put(48,16){$c$}
\put(27,55){$a$}
\put(33,84){$f$}
\thinlines

\end{picture}
&
\begin{picture}(90,80)(0,0)
\color{red}
\thicklines
\dottedline{4}(5,70)(55,70)

\color{blue}
\dashline{7}(5,30)(55,70)
\dashline{7}(5,70)(55,30)

\color{black}
\qbezier(5,30)(30,60)(55,70)
\qbezier(5,70)(30,60)(55,30)

\put(5,30){\circle*{6}}
\put(5,70){\circle*{6}}
\put(30,55){\circle*{6}}
\put(55,30){\circle*{6}}
\put(55,70){\circle*{6}}
\put(2,75){$b$}
\put(2,18){$a$}
\put(52,75){$c$}
\put(52,18){$d$}
\put(27,60){$e$}
\put(5,30){\line(0,1){40}}
\put(55,30){\line(0,1){40}}
\thinlines

\end{picture}

&
\begin{picture}(120,80)(0,0)
\thicklines
\qbezier(30,55)(55,30)(55,30)
\qbezier(30,55)(55,80)(55,80)

\put(0,55){\circle*{6}}
\put(30,55){\circle*{6}}
\put(55,30){\circle*{6}}
\put(55,80){\circle*{6}}
\put(80,55){\circle*{6}}
\put(110,55){\circle*{6}}

\put(-2,45){$a$}
\put(25,43){$b$}
\put(55,85){$c$}
\put(50,18){$d$}
\put(75,45){$e$}
\put(105,43){$f$}

\qbezier(55,30)(80,55)(80,55)
\qbezier(55,80)(80,55)(80,55)
\put(55,30){\line(0,1){50}}
\put(0,55){\line(1,0){30}}
\put(80,55){\line(1,0){30}}

\color{red}
\thicklines
\dottedline{4}(3,55)(50,80)
\dottedline{4}(60,80)(107,55)

\color{blue}
\dashline{7}(3,55)(50,30)
\dashline{7}(60,30)(107,55)
\thinlines

\end{picture}

\\
\hspace*{-1.5cm} a. {\bf f-c}$(A)$&
\hspace*{-1.5cm} 
b. {\bf f-c}$(B)$&
\hspace*{-.5cm} c. {\bf f-c}$(C)$
\end{tabular}
\end{center}
\caption{Some forbidden partial tri-colorings of small graphs. Acronym {\bf f-c} stands for {\em forbidden coloring}.}
\label{fig:forbidden2}
\end{figure}
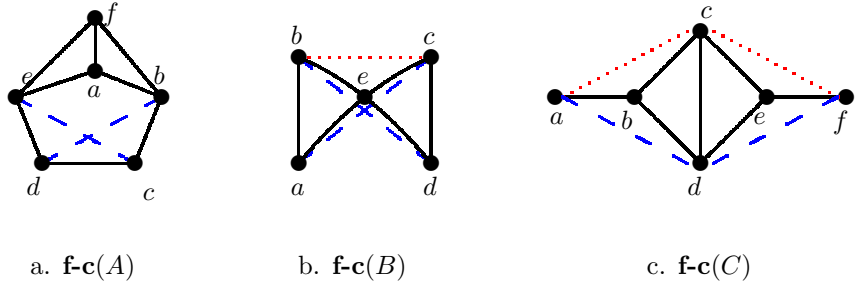

\begin{thm}
The  partial tri-colorings in Figure \ref{fig:forbidden2} are forbidden PCG-colorings.
\end{thm}

\begin{proof}
Using the results of Theorem \ref{lemma.forbidden}, we again prove separately that each tri-coloring is forbidden by contradiction.

\medskip

\noindent
{\bf Forbidden tri-coloring f-c$(A)$}:

Let us assume that the partial tri-coloring in figure \ref{fig:forbidden2}.a is a PCG-coloring. 
Consider the PCG-coloring inherited by path $G[b,c,d,e]$. 
To avoid  ${ \bf f-c}(P_4)$, non-edge $(e,b)$ must be blue. 
Now consider the PCG-coloring inherited by cycle $G[a,b,c,d,e]$. 
From Lemma \ref{lemma.cicloRossoBlu}, every PCG-coloring of $C_n$, $n \geq 4$, has at least a red non-edge. 
Thus at least one of the non-edges between $(a,c)$ and $(a,d)$ is red and w.l.o.g. let assume it is $(a,c)$. 
To avoid ${\bf f-c}(P_4)$ for path  $G[c,d,e,a]$, non-edge $(a,d)$ is red, too. 
Now, consider the PCG-coloring inherited by the  cycle $G[b,c,d,e,f]$; with a similar reasoning, 
we get that the two non-edges $(f,c)$ and $(f,d)$ are both  red. 
Thus we have four red non-edges, namely $(a,c)$, $(a,d)$, $(f,c)$ and $(f,d)$. 
This implies  ${\bf f-c}(2K_2)a$ for $G[a,c,d,f]$, a contradiction.

\medskip
\noindent
{\bf Forbidden tri-coloring f-c$(B)$}:

From the tri-coloring in Figure \ref{fig:forbidden2}.b we have that 
$$d_T(b,c) < d_{min}\leq d_T(b,e), d_T(e,c).
$$ 
Without loss of generality,  let  assume $d_T(b,e)\leq d_T(e,c)$. 
Thus $P_T(e,c)$ is the largest path in $T_{b,c,e}$. By Lemma \ref{lemma.subtree}, for leaf $d$ we have:
$d_T(b,d) \leq \max\left\{d_T(d,e), d_T(c,d)\right\}$ while from the tri-coloring it holds that 
$d_T(d,e), d_T(c,d) \leq d_{max} < d_T(b,d)$, a contradiction.

\medskip
\noindent
{\bf Forbidden tri-coloring f-c$(C)$}:

From the the tri-coloring in Figure \ref{fig:forbidden2}.c, extract  the inherited $PCG$-colorings for the two subgraphs $G[a,c,d,e]$ and $G[b,c,d,f]$.
To avoid {\bf f-c}$(K_3 \cup K_1)$, 
the non-edges $(a,e)$ and $(b,f)$ are both blue.
Now we distinguish the two possible cases for the color of non-edge $(a,f)$:
\begin{description}
\item[$(a,f)$] is red: 
consider  the $PCG$-coloring for subgraph $G[a,b,e,f]$. To avoid {\bf f-c}$(2K_2)b$, non-edge $(b,e)$ has to be blue. This implies that the 
$PCG$-coloring for path $G[a,b,d,e,f]$ has  all the 2-non-edges with color blue while the non-edge $(a,f)$ is red. This is in contradiction with Lemma \ref{le1}.
\item[$(a,f)$] is blue: 
in this case consider Lemma \ref{lemma.subtree} applied to tree $T_{a,d,f}$.
We distinguish the three cases for the largest path among $P_T(a,d)$, $P_T(a,f)$ and $P_T(d,f)$:
\begin{description}
\item[$P_T(a,d):$] for leaf $b$ it must be: $d_T(f,b) \leq \max\left\{d_T(a,b), d_T(d,b)\right\}$ while from the tri-coloring $d_T(a,b), d_T(d,b)\leq d_{max}<d_T(f,b)$.
\item[$P_T(a,f):$] for leaf $c$ it must be:  $d_T(d,c) \leq \max\left\{d_T(a,c), d_T(f,c)\right\}$ while from the tri-coloring $d_T(a,c), d_T(f,c)< d_{min}\leq d_T(d,c)$.
\item[$P_T(d,f):$] for leaf $e$ it must be:  $d_T(a,e) \leq \max\left\{d_T(d,e), d_T(f,e)\right\}$ while from the tri-coloring $d_T(d,e), d_T(f,e)\leq d_{max}< d_T(a,e)$.
\end{description}
In all the three cases, a contradiction arises.
\end{description}
\end{proof}

\section{Wheels}
\label{sec.wheel}

Wheels $W_{n+1}$ are $n$ length cycles $C_n$ whose nodes are all connected with a universal node.
They have already been studied from the pairwise compatibility point of view.
Indeed,  wheel $W_{6+1}$ is PCG and it is the only graph with 7 nodes whose witness tree is not a caterpillar \cite{CFS13} (see Figure \ref{fig:wheel}.a).
Moreover, it has been proven in \cite{CAS13} that also the larger wheels up to $W_{10+1}$ do not have a caterpillar as a witness tree but, up to now, no other witness trees are known for these graphs and, in general, it has been left as an open problem whether wheels with at least 8 nodes are PCGs or not.
In this section we completely solve this problem.

\begin{figure}[t]
\vspace*{1cm}
\begin{center}
\begin{tabular}{c  c }
\begin{picture}(100,80)(0,0)
\hspace*{-1.5cm}
\put(10,10){\circle*{6}}
\put(10,40){\circle*{6}}
\put(10,60){\circle*{6}}
\put(10,90){\circle*{6}}

\put(40,25){\circle*{6}}
\put(40,75){\circle*{6}}

\put(70,50){\circle*{6}}

\put(100,25){\circle*{6}}
\put(100,75){\circle*{6}}

\put(130,10){\circle*{6}}
\put(130,40){\circle*{6}}

\put(-4,10){$v_3$}
\put(-4,40){$v_6$}
\put(-4,60){$v_5$}
\put(-4,90){$v_2$}
\put(103,75){$c$}
\put(132,11){$v_1$}
\put(132,41){$v_4$}

\put(20,6){$1$}
\put(20,38){$3$}
\put(20,56){$1$}
\put(20,88){$3$}

\put(48,38){$1$}
\put(48,56){$1$}

\put(88,38){$1$}
\put(88,56){$3$}

\put(116,6){$1$}
\put(116,38){$3$}

\thicklines
\qbezier(10,10)(40,25)(40,25)
\qbezier(10,40)(40,25)(40,25)
\qbezier(10,60)(40,75)(40,75)
\qbezier(10,90)(40,75)(40,75)
\qbezier(40,25)(70,50)(70,50)
\qbezier(40,75)(70,50)(70,50)
\qbezier(70,50)(100,25)(100,25)
\qbezier(70,50)(100,75)(100,75)
\qbezier(100,25)(130,10)(130,10)
\qbezier(100,25)(130,40)(130,40)
\end{picture}
&
\begin{picture}(100,80)(0,0)
\hspace*{.5cm}
\put(10,50){\circle*{6}}
\put(10,80){\circle*{6}}

\put(40,65){\circle*{6}}

\put(60,65){\circle*{6}}

\put(80,65){\circle*{6}}


\put(110,50){\circle*{6}}
\put(110,80){\circle*{6}}

\put(40,90){\circle*{6}}
\put(80,90){\circle*{6}}

\put(60,40){\circle*{6}}

\put(45,10){\circle*{6}}
\put(75,10){\circle*{6}}

\put(-4,50){$v_2$}
\put(-4,80){$v_4$}
\put(115,50){$v_1$}
\put(115,80){$v_3$}
\put(40,0){$c$}
\put(70,0){$v_5$}
\put(35,100){$v_7$}
\put(75,100){$v_6$}

\put(43,20){$3$}
\put(70,20){$6$}
\put(61,50){$3$}

\put(45,67){$2$}
\put(66,67){$2$}
\put(41,78){$5$}
\put(72,78){$1$}
\put(25,49){$3$}
\put(25,75){$1$}
\put(87,49){$3$}
\put(87,75){$5$}

%
%

\thicklines
\qbezier(10,50)(40,65)(40,65)
\qbezier(10,80)(40,65)(40,65)
\qbezier(110,50)(80,65)(80,65)
\qbezier(110,80)(80,65)(80,65)
\qbezier(45,10)(60,40)(60,40)
\qbezier(75,10)(60,40)(60,40)

\put(40,65){\line(1,0){40}}
\put(40,65){\line(0,1){25}}
\put(80,65){\line(0,1){25}}
\put(60,65){\line(0,-1){25}}
\end{picture}
\\
\hspace*{-1.3cm}a.
&
\hspace*{3.2cm}b. 
\end{tabular}
\end{center}
\caption{a. Tree $T$ such that $W_{6+1}=PCG(T, 5,7)$; b.  Tree $T$ such that $W_{7+1}=PCG(T, 9,13)$.}
\label{fig:wheel}
\end{figure}
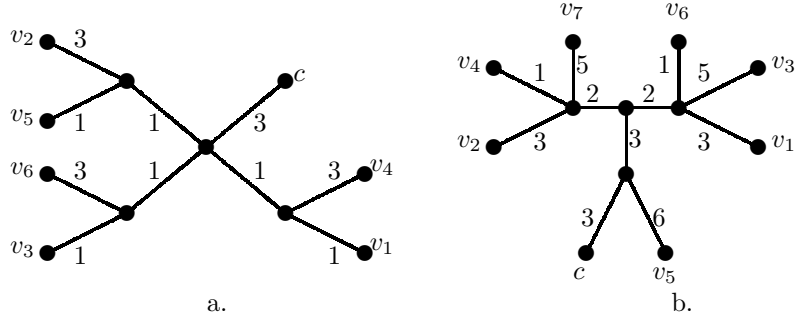

First we prove that $W_{7+1}$ is PCG.

\begin{thm}
Wheel $W_{7+1}$ is PCG.
\end{thm}

\begin{proof}
In order to prove the statement, it is enough to show a triple $(T,$ $d_{min},$ $d_{max})$ witnessing that $W_{7+1}$ is PCG.
Tree $T$ is shown in Figure \ref{fig:wheel}.b, and the values of $d_{min}$ and $d_{max}$ are 9 and 13, respectively.
\end{proof}

%

Then, exploiting the proof technique just described, we prove that every larger wheel $W_{n+1}, n \geq 8$,  is not a PCG.

\begin{thm}
\label{th.wheel}
Let $n\geq 8$. The graph $W_{n+1}$ is not PCG.
\end{thm}
\begin{proof}
Step 1 of the proof technique, requiring a list of useful forbidden PCG-colorings, has been completed in Section \ref{sec.forbidden}: namely, we will use {\bf f-c}($2K_2$)a, {\bf f-c}($P_4$), {\bf f-c}($K_{1,3}$), {\bf f-c}($B$) and the forbidden tri-coloring in Theorem \ref{cicloDistanza2}.

Step 2 of the proof technique requires to prove that every tri-coloring of $W_{n+1}$ induces a forbidden PCG-coloring for a certain induced pairwise compatibility subgraph.

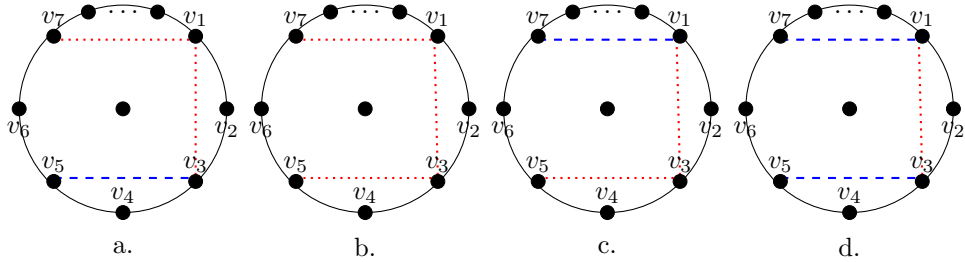
\begin{figure}[ht]
\vspace*{0.5cm}
\begin{tikzpicture}[scale=0.92]
\draw (1,2) circle (1.5cm);
\draw[thick,color=red,dotted] (0,3) -- (2,3);
\draw[thick,color=blue,dashed] (0,1) -- (2,1);
\draw[thick,color=red,dotted] (2.05,3) -- (2.05,1);
\draw[fill=black] (0,0.95) circle (0.1cm); \node[] at (0,1.2) {$v_5$};
\draw[fill=black] (0,3.05) circle (0.1cm); \node[] at (0,3.3) {$v_7$};
\draw[fill=black] (2.05,3.05) circle (0.1cm); \node[] at (2.05,3.3) {$v_1$};
\draw[fill=black] (2.05,0.95) circle (0.1cm); \node[] at (2.05,1.2) {$v_3$};
\draw[fill=black] (2.5,2) circle (0.1cm); \node[] at (2.5,1.7) {$v_2$};
\draw[fill=black] (1,0.5) circle (0.1cm); \node[] at (1,.75) {$v_4$};
\draw[fill=black] (-0.5,2) circle (0.1cm);  \node[] at (-.5,1.7) {$v_6$};
\draw[fill=black] (0.5,3.4) circle (0.1cm); \node[] at (1,3.4) {$\ldots$};
\draw[fill=black] (1.5,3.4) circle (0.1cm); 
\draw[fill=black] (1,2) circle (0.1cm);
\draw (4.5,2) circle (1.5cm);
\draw[thick,color=red,dotted] (3.5,3) -- (5.5,3);
\draw[thick,color=red,dotted] (3.5,1) -- (5.5,1);
\draw[thick,color=red,dotted] (5.5,3) -- (5.55,1);
\draw[fill=black] (3.5,0.95) circle (0.1cm); \node[] at (3.5,1.2) {$v_5$};
\draw[fill=black] (3.5,3.05) circle (0.1cm); \node[] at (3.5,3.3) {$v_7$};
\draw[fill=black] (5.55,3.05) circle (0.1cm); \node[] at (5.55,3.3) {$v_1$};
\draw[fill=black] (5.55,0.95) circle (0.1cm); \node[] at (5.55,1.2) {$v_3$};
\draw[fill=black] (6,2) circle (0.1cm); \node[] at (6,1.7) {$v_2$};
\draw[fill=black] (4.5,0.5) circle (0.1cm);  \node[] at (4.5,.75) {$v_4$};
\draw[fill=black] (3,2) circle (0.1cm);  \node[] at (3,1.7) {$v_6$};
\draw[fill=black] (4,3.4) circle (0.1cm); \node[] at (4.5,3.4) {$\ldots$};
\draw[fill=black] (5,3.4) circle (0.1cm); 
\draw[fill=black] (4.5,2) circle (0.1cm);
\draw (8,2) circle (1.5cm);
\draw[thick,color=blue,dashed] (7,3) -- (9,3);
\draw[thick,color=red,dotted] (7,1) -- (9,1);
\draw[thick,color=red,dotted] (9,3) -- (9.05,1);
\draw[fill=black] (7,0.95) circle (0.1cm); \node[] at (7,1.2) {$v_5$};
\draw[fill=black] (7,3.05) circle (0.1cm); \node[] at (7,3.3) {$v_7$};
\draw[fill=black] (9.05,3.05) circle (0.1cm); \node[] at (9.05,3.3) {$v_1$};
\draw[fill=black] (9.05,0.95) circle (0.1cm); \node[] at (9.05,1.2) {$v_3$};
\draw[fill=black] (9.5,2) circle (0.1cm); \node[] at (9.5,1.7) {$v_2$};
\draw[fill=black] (8,0.5) circle (0.1cm);  \node[] at (8,.75) {$v_4$};
\draw[fill=black] (6.5,2) circle (0.1cm);  \node[] at (6.5,1.7) {$v_6$};
\draw[fill=black] (7.5,3.4) circle (0.1cm); \node[] at (8,3.4) {$\ldots$};
\draw[fill=black] (8.5,3.4) circle (0.1cm); 
\draw[fill=black] (8,2) circle (0.1cm);
\draw (11.5,2) circle (1.5cm);
\draw[thick,color=blue,dashed] (10.5,3) -- (12.5,3);
\draw[thick,color=blue,dashed] (10.5,1) -- (12.5,1);
\draw[thick,color=red,dotted] (12.5,3) -- (12.55,1);
\draw[fill=black] (10.5,0.95) circle (0.1cm); \node[] at (10.5,1.2) {$v_5$};
\draw[fill=black] (10.5,3.05) circle (0.1cm); \node[] at (10.5,3.3) {$v_7$};
\draw[fill=black] (12.55,3.05) circle (0.1cm); \node[] at (12.55,3.3) {$v_1$};
\draw[fill=black] (12.55,0.95) circle (0.1cm); \node[] at (12.55,1.2) {$v_3$};
\draw[fill=black] (13,2) circle (0.1cm); \node[] at (13,1.7) {$v_2$};
\draw[fill=black] (11.5,0.5) circle (0.1cm);  \node[] at (11.5,.75) {$v_4$};
\draw[fill=black] (10,2) circle (0.1cm);  \node[] at (10,1.7) {$v_6$};
\draw[fill=black] (11,3.4) circle (0.1cm); \node[] at (11.5,3.4) {$\ldots$};
\draw[fill=black] (12,3.4) circle (0.1cm); 
\draw[fill=black] (11.5,2) circle (0.1cm);

\node[] at (1,0) {a.};
\node[] at (4.5,0) {b.};
\node[] at (8,0) {c.};
\node[] at (11.5,0) {d.};
\end{tikzpicture}
\vspace*{-0.5cm}
\caption{The four cases in the proof of Theorem \ref{th.wheel}.}
\label{fig:proof_wheel}
\end{figure}

Let be given any tri-coloring of $W_{n+1}$; in view of Theorem \ref{cicloDistanza2}, there exists a red 2-non-edge
, w.l.o.g.  let it be $(v_1, v_3)$.
Let us now consider the three non-edges $(v_7, v_1)$, $(v_1, v_3)$, $(v_3, v_5)$. There are only 4 possibilities for the colors of these non-edges and we will study them one by one (see Figure \ref{fig:proof_wheel}).

\medskip
\noindent {\bf Case in Figure \ref{fig:proof_wheel}.a}:

Assume first that $(v_4, v_7)$ is blue; then  
non-edge $(v_3, v_7)$ is necessarily red in order to avoid {\bf f-c($K_{1,3}$)} on the graph induced by nodes $c$, $v_1$, $v_3$ and $v_7$. In the following we summarize this sentence as: 

$(v_3, v_7)$ red $\leftarrow$ {\bf f-c($K_{1,3}$)} on $G[c, v_1, v_3, v_7]$.

\noindent
and  a chain of obliged colored non-edges follows, namely: 

\begin{ditemize}
\item
$(v_3, v_6)$ red $\leftarrow$ {\bf f-c($B$)} on $G[c, v_3, v_4, v_6, v_7]$ (indeed, $(v_3, v_7)$ is red and $(v_4, v_7)$ is blue, so $(v_3, v_6)$ cannot be blue)
\item
$(v_1, v_4)$ blue $\leftarrow$ {\bf f-c($K_{1,3}$)} on $G[c, v_1, v_4, v_7]$
\item
$(v_1, v_6)$ red $\leftarrow$ {\bf f-c($K_{1,3}$)} on $G[c, v_1, v_3, v_6]$
\item
$(v_4, v_6)$ blue $\leftarrow$ {\bf f-c($K_{1,3}$)} on $G[c, v_1, v_4, v_6]$
\end{ditemize}
We got a path induced by nodes $v_3$, $v_4$, $v_5$ and $v_6$ with forbidden coloring {\bf f-c($P_4$)}, a contradiction, meaning that $(v_4, v_7)$ cannot be blue.

So, $(v_4, v_7)$ is red, and
we have the following chain of obliged colored non-edges: 

\begin{ditemize}
\item
$(v_1, v_5)$ blue $\leftarrow$ {\bf f-c($K_{1,3}$)} on $G[c, v_1, v_3, v_5]$
\item
$(v_1, v_4)$ red $\leftarrow$ {\bf f-c($K_{1,3}$)} on $G[c, v_1, v_4, v_7]$
\item
$(v_2, v_4)$ red $\leftarrow$ {\bf f-c($B$)} on $G[c, v_1, v_2, v_4, v_5]$
\item
$(v_2, v_7)$ red $\leftarrow$ {\bf f-c($K_{1,3}$)} on $G[c, v_2, v_4, v_7]$
\item
$(v_5, v_7)$ blue $\leftarrow$ {\bf f-c($K_{1,3}$)} on $G[c, v_1, v_5, v_7]$
\item
$(v_4, v_6)$ red $\leftarrow$ {\bf f-c($P_4$)} on $G[v_4, v_5, v_6, v_7]$
\item
$(v_2, v_6)$ red $\leftarrow$ {\bf f-c($K_{1,3}$)} on $G[c, v_2, v_4, v_6]$
\item
$(v_1, v_6)$ red $\leftarrow$ {\bf f-c($K_{1,3}$)} on $G[c, v_1, v_4, v_6]$
\end{ditemize}

Graph $G[v_1, v_2, v_6, v_7]$ has forbidden coloring {\bf f-c($2K_2$)}a, and this is a contradiction, meaning that $(v_4, v_7)$ cannot be red.

\medskip
\noindent {\bf Case in Figure \ref{fig:proof_wheel}.b}:

Notice that:
\begin{ditemize}
\item

$(v_3, v_7)$ red $\leftarrow$ {\bf f-c($K_{1,3}$)} on $G[c, v_1, v_3, v_7]$
\item
$(v_5, v_7)$ red $\leftarrow$ {\bf f-c($K_{1,3}$)} on $G[c, v_3, v_5, v_7]$
\item
$(v_1, v_5)$ red $\leftarrow$ {\bf f-c($K_{1,3}$)} on $G[c, v_1, v_3, v_5]$
\end{ditemize}

Assume now that $(v_4, v_7)$ is blue; then
we have the following chain of oblied colored non-edges:

\begin{ditemize}
\item
$(v_5, v_8)$ red $\leftarrow$ {\bf f-c($B$)} on $G[c, v_4, v_5, v_7, v_8]$
\item
$(v_3, v_8)$ red $\leftarrow$ {\bf f-c($K_{1,3}$)} on $G[c, v_3, v_5, v_8]$
\item
$(v_1, v_4)$ blue $\leftarrow$ {\bf f-c($K_{1,3}$)} on $G[c, v_1, v_4, v_7]$
\item
$(v_2, v_5)$ red $\leftarrow$ {\bf f-c($B$)} on $G[c, v_1, v_2, v_4, v_5]$
\item
$(v_2, v_8)$ red $\leftarrow$ {\bf f-c($K_{1,3}$)} on $G[c, v_2, v_5, v_8]$
\item
$(v_2, v_7)$ red $\leftarrow$ {\bf f-c($K_{1,3}$)} on $G[c, v_2, v_5, v_7]$
\end{ditemize}

so $G[v_2, v_3, v_7, v_8]$ has forbidden coloring {\bf f-c($2K_2$)}a, a contradiction.

So, $(v_4, v_7)$ must be red, and
$(v_1, v_4)$ red $\leftarrow$ {\bf f-c($K_{1,3}$)} on $G[c, v_1, v_4, v_7]$.

Now, we consider the non-edge $(v_1, v_6)$.
If $(v_1, v_6)$ is red:

\begin{ditemize}
\item
$(v_4, v_6)$ red $\leftarrow$ {\bf f-c($K_{1,3}$)} on $G[c, v_1, v_4, v_6]$
\item
$(v_3, v_6)$ red $\leftarrow$ {\bf f-c($K_{1,3}$)} on $G[c, v_1, v_3, v_6]$
\end{ditemize}

and we have a contradiction arisen from having {\bf f-c$(2K_2)$}a on $G[v_3, v_4, v_6, v_7]$.

If, on the contrary, $(v_1, v_6)$ is blue, then:

\begin{ditemize}
\item
$(v_2, v_7)$ red $\leftarrow$ {\bf f-c($B$)} on $G[c, v_1, v_2, v_6, v_7]$
\item
$(v_2, v_4)$ red $\leftarrow$ {\bf f-c($K_{1,3}$)} on $G[c, v_2, v_4, v_7]$
\item
$(v_2, v_5)$ red $\leftarrow$ {\bf f-c($K_{1,3}$)} on $G[c, v_2, v_5, v_7]$
\end{ditemize}
deducing a contradiction on $G[v_1, v_2, v_4, v_5]$ with forbidden coloring {\bf f-c$(2K_2)$}a.

\noindent {\bf Case in Figure \ref{fig:proof_wheel}.c}:

\begin{ditemize}
\item
$(v_3, v_7)$ blue $\leftarrow$ {\bf f-c($K_{1,3}$)} on $G[c, v_1, v_3, v_7]$
\item
$(v_5, v_7)$ blue $\leftarrow$ {\bf f-c($K_{1,3}$)} on $G[c, v_3, v_5, v_7]$
\end{ditemize}

Let us now consider in this order the non-edges $(v_5, v_n)$, $(v_5, v_{n-1}), \ldots$ and let $(v_5, v_i)$ be the first encountered blue non-edge, surely existing because $(v_5, v_7)$ is blue.\\
We distinguish two subcases: either $i=n$ or $i<n$.

If $i=n$:
\begin{ditemize}
\item
$(v_3, v_n)$ blue $\leftarrow$ {\bf f-c($K_{1,3}$)} on $G[c, v_3, v_5, v_n]$
\item
$(v_1, v_5)$ red $\leftarrow$ {\bf f-c($K_{1,3}$)} on $G[c, v_1, v_3, v_5]$
\item
$(v_1, v_6)$ red $\leftarrow$ {\bf f-c($B$)} on $G[c, v_n, v_1, v_5, v_6]$
\item
$(v_3, v_6)$ red $\leftarrow$ {\bf f-c($K_{1,3}$)} on $G[c, v_1, v_3, v_6]$
\item
$(v_6, v_n)$ blue $\leftarrow$ {\bf f-c($K_{1,3}$)} on $G[c, v_3, v_6, v_n]$
\end{ditemize}
Now, If $n=8$, then $v_7$ and $v_n$ are adjacent and $G[v_6, v_7, v_n, v_1]$ has forbidden tri-coloring {\bf f-c$(P_4)$}.
If, on the contrary, $n>8$, then we have the forbidden tri-coloring  {\bf f-c$(B)$} on $G[c, v_n, v_1, v_6, v_7]$.

If $i<n$,
we know that $(v_5, v_{i+1})$ is red; moreover:
\begin{ditemize}
\item
$(v_3, v_{i+1})$ red $\leftarrow$ {\bf f-c($K_{1,3}$)} on $G[c, v_3, v_5, v_{i+1}]$
\item
$(v_3, v_i)$ blue $\leftarrow$ {\bf f-c($K_{1,3}$)} on $G[c, v_3, v_5, v_i]$
\item
$(v_2, v_{i+1})$ red $\leftarrow$ {\bf f-c($B$)} on $G[c, v_2, v_3, v_i, v_{i+1}]$
\item
$(v_2, v_5)$ red $\leftarrow$ {\bf f-c($K_{1,3}$)} on $G[c, v_2, v_5, v_{i+1}]$
\item
$(v_6, v_{i+1})$ red $\leftarrow$ {\bf f-c($B$)} on $G[c, v_5, v_6, v_i, v_{i+1}]$
\item
$(v_2, v_6)$ red $\leftarrow$ {\bf f-c($K_{1,3}$)} on $G[c, v_2, v_6, v_{i+1}]$
\item
$(v_3, v_6)$ red $\leftarrow$ {\bf f-c($K_{1,3}$)} on $G[c, v_3, v_6, v_{i+1}]$
\end{ditemize}
We get subgraph $G[v_2, v_3, v_5, v_6]$ colored with {\bf f-c$(2K_2)$}a.

\noindent {\bf Case in Figure \ref{fig:proof_wheel}.d}:

We distinguish two subcases, according to the color of non-edge $(v_1, v_4)$.

If $(v_1, v_4)$ is blue:

\begin{ditemize}
\item
$(v_3, v_n)$ red $\leftarrow$ {\bf f-c($B$)} on $G[c, v_n, v_1, v_3, v_4]$
\item
$(v_5, v_n)$ blue $\leftarrow$ {\bf f-c($K_{1,3}$)} on $G[c, v_3, v_5, v_n]$
\item
$(v_4, v_n)$ blue $\leftarrow$ {\bf f-c($B$)} on $G[c, v_n, v_1, v_4,v_5]$
\item
$(v_3, v_7)$ blue $\leftarrow$ {\bf f-c($K_{1,3}$)} on $G[c, v_1, v_3, v_7]$
\end{ditemize}
Now we show that $(v_3, v_n)$ red and $(v_4, v_n)$ blue imply $(v_3, v_8)$ red and $(v_4, v_8)$ blue,
so obtaining $G[c,v_3, v_4, v_7, v_8]$ with forbidden coloring {\bf f-c($B$)}, a contradiction.\\
To show the assert it is sufficient to prove that if $(v_3, v_i)$ is red and $(v_4, v_i)$ is blue and  $i>8$, then $(v_3, v_{i-1})$ is red and $(v_4, v_{i-1})$ is blue.
\begin{ditemize}
\item
$(v_3, v_{i-1})$ red $\leftarrow$ {\bf f-c($B$)} on $G[c, v_3, v_4, v_{i-1}, v_i]$
\item
$(v_1, v_{i-1})$ red $\leftarrow$ {\bf f-c($K_{1,3}$)} on $G[c, v_1, v_3, v_{i-1}]$
\item
$(v_4, v_{i-1})$ blue $\leftarrow$ {\bf f-c($K_{1,3}$)} on $G[c, v_1, v_4,v_{i-1}]$
\end{ditemize}
and this part of the proof is concluded.

If, instead, $(v_1, v_4)$ is red:
\begin{ditemize}
\item
$(v_4, v_7)$ blue $\leftarrow$ {\bf f-c($K_{1,3}$)} on $G[c,  v_1, v_4, v_7]$
\item
$(v_1, v_5)$ blue $\leftarrow$ {\bf f-c($K_{1,3}$)} on $G[c, v_1, v_3, v_5]$
\item
$(v_4, v_n)$ red $\leftarrow$ {\bf f-c($B$)} on $G[c, v_n, v_1, v_4,v_5]$
\item
$(v_3, v_n)$ blue $\leftarrow$  {\bf f-c$(2K_2)$}a on $G[ v_1, v_n, v_3, v_4]$
\end{ditemize}
Now, if $n=8$ then the nodes $v_7$ and $v_8$ are adjacent and $G[v_3, v_4, v_7, v_n]$ has forbidden tri-coloring {\bf f-c($B$)}.
Thus, let us assume  $n>8$. 
\begin{ditemize}
\item
$(v_4, v_{n-1})$ red $\leftarrow$ {\bf f-c($B$)} on $G[c,  v_3, v_4, v_{n-1},v_n]$
\item
$(v_1, v_{n-1})$ red $\leftarrow$  {\bf f-c($K_{1,3}$)} on $G[c, v_1, v_4, v_{n-1}]$
\item
$(v_3, v_{n-1})$ red $\leftarrow$ {\bf f-c($K_{1,3}$)} on $G[c,  v_1, v_3,v_{n-1}]$
\item
$(v_5, v_{n-1})$ blue $\leftarrow$  {\bf f-c($K_{1,3}$)} on $G[ c, v_3, v_5,v_{n-1}]$
\end{ditemize}
Similarly to what we did before, now we show that $(v_4, v_{n-1})$ red and $(v_5, v_{n-1})$ blue imply $(v_4, v_8)$ red and $(v_5, v_8)$ blue, so obtaining $G[c,v_4, v_5, v_7, v_8]$ with forbidden coloring {\bf f-c($B$)}, a contradiction.\\
To show the assert it is sufficient to prove that if $(v_4, v_i)$ is red and $(v_5, v_i)$ is blue and  $i>8$, then $(v_4, v_{i-1})$ is red and $(v_5, v_{i-1})$ is blue.
\begin{ditemize}
\item
$(v_4, v_{i-1})$ red $\leftarrow$ {\bf f-c($B$)} on $G[c, v_4, v_5, v_{i-1}, v_i]$
\item
$(v_1, v_{i-1})$ red $\leftarrow$ {\bf f-c($K_{1,3}$)} on $G[c, v_1, v_4, v_{i-1}]$
\item
$(v_5, v_{i-1})$ blue $\leftarrow$ {\bf f-c($K_{1,3}$)} on $G[c, v_1, v_5,v_{i-1}]$
\end{ditemize}

Step 3 of the proof technique: we deduce that $G$ is not PCG since all the partial colorings shown in Figure \ref{fig:proof_wheel} are not feasible.
\end{proof}

\section{The strong product of a cycle and $P_2$}
\label{ref.product}

\begin{figure}[t]
\hspace*{4cm}
\begin{tikzpicture}[scale=0.85]
\draw (4,4) circle (2cm);
\draw (4,4) circle (3cm);

\draw[fill=black] (2.95,5.7) circle (0.1cm); \node[] at (3.2,5.4) {$v_n$};
\draw[fill=black] (2.4,6.52) circle (0.1cm); \node[] at (2.3,6.9) {$u_n$};
\draw[] (2.95,5.7) -- (2.4,6.52);
\draw[dashed] (2.95,5.7) -- (2.4,5.8);
\draw[dashed] (2.4,6.52) -- (2.3,5.9);

\draw[fill=black] (4,6) circle (0.1cm); \node[] at (4,5.6) {$v_1$};
\draw[fill=black] (4,7) circle (0.1cm); \node[] at (4,7.4) {$u_1$};
\draw[] (4,6) -- (4,7);

\draw[fill=black] (5.05,5.7) circle (0.1cm); \node[] at (4.8,5.4) {$v_2$};
\draw[fill=black] (5.6,6.52) circle (0.1cm); \node[] at (5.7,7) {$u_2$};
\draw[] (5.05,5.7) -- (5.6,6.52);

\draw[fill=black] (5.78,4.95) circle (0.1cm); \node[] at (5.45,4.8) {$v_3$};
\draw[fill=black] (6.65,5.5) circle (0.1cm); \node[] at (7.1,5.6) {$u_3$};
\draw[] (5.78,4.95) -- (6.65,5.5);

\draw[fill=black] (6,4) circle (0.1cm); \node[] at (5.6,4) {$v_4$};
\draw[fill=black] (7,4) circle (0.1cm); \node[] at (7.4,4) {$u_4$};
\draw[] (6,4) -- (7,4);
\draw[dashed] (6,4) -- (6.4,3.5);
\draw[dashed] (7,4) -- (6.6,3.5);

\node[] at (7.3,3) {$\ldots$};
\node[] at (5.4,3) {$\ldots$};
\node[] at (1.7,6.3) {$\ldots$};
\node[] at (2.7,4.7) {$\ldots$};

\draw[] (2.4,6.52) -- (4,6); \draw[] (2.95,5.7) -- (4,7);
\draw[] (4,7) -- (5.05,5.7); \draw[] (4,6) -- (5.6,6.52);
\draw[] (5.05,5.7) -- (6.65,5.5); \draw[] (5.78,4.95) -- (5.6,6.52);
\draw[] (5.78,4.95) -- (7,4); \draw[] (6,4) -- (6.65,5.5);

\end{tikzpicture}
\caption{Graph $C_n \square P_2$.}
\label{fig:CnK2}
\end{figure}
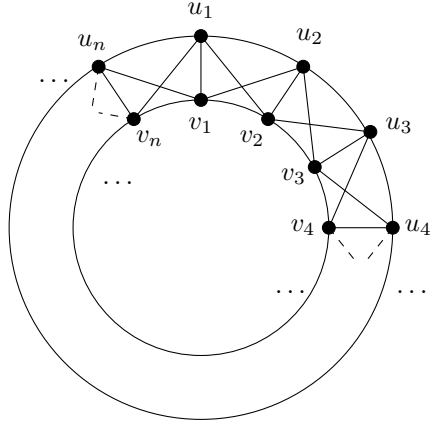

Given two graphs $G$ and $H$, their {\em strong product} $G \square H$ is a graph whose node set is the cartesian product of the node sets of the two graphs, and there is an edge between nodes $(u, v)$ and $(u', v')$ if and only if either $u=u'$ and $(v,v')$ is an edge of $H$ or $v=v'$ and $(u,u')$ is an edge of $G$.

In the following, we study graph $C_n \square P_2$, a $2n$ node graph in which two cycles are naturally highlighted; we call $v_1,\ldots, v_n$ and $u_1,u_2,\ldots, u_n$, respectively, their nodes as shown in Figure \ref{fig:CnK2}.

We recall that $C_4 \square P_2$, i.e. the graph depicted in Figure \ref{fig.notPCG}.b, has already been proved not to be PCG \cite{DMR15}.

We apply our technique to $C_n \square P_2$, by showing that every tri-coloring leads to forbidden tri-coloring {\bf f-c}$(C)$.
Since this tri-coloring appears only when $n \geq 6$, we need to handle the case  $C_5 \square P_2$ separately.

\begin{thm}
Graph  $C_5 \square P_2$ is not PCG.
\end{thm}

\begin{proof}
According to the second step of the proof technique, we focus on any tri-coloring of $C_5 \square P_2$ and prove that it is forbidden.

Consider cycle $G[v_1,v_2,v_3,v_4,v_5]=PCG(T, d_{min}, d_{max})$; from Lemma \ref{lemma.cicloRossoBlu}, there exists at least a blue non-edge.

Thus, w.l.o.g. assume that non-edge $(v_2,v_5)$ is blue. In order to avoid forbidden coloring ${\bf f-c}(A)$ on the induced subgraph $G[v_1,v_2,v_3,u_4,v_4,v_5]$, 
non-edge $(v_1,v_4)$ must be red. 
The same reasoning can be used for the  following three induced subgraphs: $G[u_1,v_2,v_3,u_3,u_4,v_5]$, $G[u_1,v_2,v_3,v_4,u_4,v_5]$ and $G[u_1,v_2,v_3,v_4,u_4,v_5]$ to prove that non-edges $(u_1,v_4)$,  $(v_1,v_3)$ and $(u_1,v_3)$ must be red, too.
We get ${ \bf f-c}(2K_2)a$  on the induced subgraph $G[u_1,v_1,v_3,v_4]$, a contradiction.

In view of the last step of the proof technique, $C_5 \square P_2$ is not PCG, so concluding the proof.
\end{proof}

\begin{thm}
Graph $C_n \square P_2$, $n\geq 6$, is not PCG.
\end{thm}
\begin{proof}
We exploit again the technique described in Section \ref{sec.technique}.

For Step 1,  we will use {\bf f-c}($2K_2$)a, {\bf f-c}($K_3 \cup K_1$), {\bf f-c}($B$), {\bf f-c}($C$) and the forbidden tri-coloring in Theorem \ref{cicloDistanza2}.

According to Step 2, 
we prove that for each tri-coloring of  $C_n \square P_2$, with $n \geq 6$, there exists an induced subgraph of $C_n \square P_2$ that inherits a forbidden PCG-coloring.

Let fix any tri-coloring of $C_n \square P_2$. 
Consider the cycle $G[v_1,v_2,\ldots,v_n]$; in view of Theorem \ref{cicloDistanza2}, there exists a red 2-non-edge in the cycle, w.l.o.g. let it be $(v_2,v_4)$. 
Consider now the induced subgraph $G[v_2,u_2,v_3,v_4,u_4]$.
In order to avoid {\bf f-c($B$)}, at least one between the non-edges $(u_{2}, v_4)$ and $(v_2, u_{4})$ must be red.
Thus, either $(v_2, v_{4})$ and $(u_2, v_{4})$ are  red or $(v_{2}, v_4)$ and $(v_{2}, u_4)$ are  red. 
Due to the symmetry of $C_n \square P_2$, it is not restrictive to assume that non-edges $(v_2,v_4)$ and $(v_2,u_4)$ are red. 
From this, we can prove that all the non-edges incident on $v_2$ are red.  To do that, it is sufficient to show that if  non-edges $(v_2,v_i)$ and $(v_2,u_i)$, $4\leq i<n$, are red,  then non-edges  $(v_2,v_{i+1})$ and $(v_2,u_{i+1})$ are red, too. 
To this aim consider  the induced subgraph $G[v_2,v_i,u_i,v_{i+1}]$; in order to avoid {\bf f-c$(K_3 \cup K_1)$}, on the three non-edges $(v_2,v_i)$, $(v_2,u_i)$ and $(v_2,v_{i+1})$ the red color can not appear exactly twice. 
Since $(v_2,v_i)$ and  $(v_2,u_i)$ are both red, it follows that $(v_2,v_{i+1})$ must also be red. 
Analogously, considering the induced subgraph $G[v_2,v_i,u_i,u_{i+1}]$, to avoid {\bf f-c$(K_3 \cup K_1)$} we get that $(v_2,u_{i+1})$ is red.

In particular, when $i=n-1$, we have that $(v_2,v_n)$ and $(v_2,u_n)$ are both red. 
Consider now the induced subgraph $G[v_2,u_2,v_n,u_n]$; to avoid {\bf f-c$(2K_2)$}a, we have that $(u_2,x)$, with $x\in \{u_n,v_n\}$, must be a blue non-edge. Analogously, to avoid {\bf f-c$(2K_2)$}a on the induced graph $G[v_2, u_2, v_4, u_4]$,  $(u_2,y)$, with $y\in \{u_4,v_4\}$ must be a blue non edge.
Finally, we get the  {\bf f-c$(C)$} on the induced graph $G[x, v_1,  v_2,  u_2,  v_3, y]$, a contradiction.

Step 3 of the proof technique concludes the proof.
\end{proof}

\section{Minimality}
\label{sec.minimality}

If a graph contains as induced subgraph a not PCG, then it is not PCG, too.
We call {\em minimal non PCG} a graph that is not PCG and it does not contain any proper induced subgraph that is not PCG.

In this section we prove that all graphs inside each one of the two considered classes we have just proved not to be PCGs are minimal not PCGs. More in detail, we prove that by deleting any node from the considered graph, we get a PCG.

\medskip

The following theorem states that wheels are minimal not PCGs.
\begin{thm}
Let $n\geq 8$. The graph obtained by removing any node from $W_{n+1}$ is PCG. In other words, $W_{n+1}$ is a minimal not PCG.
\end{thm}

\begin{proof}
Notice that, if we remove from $W_{n+1}$ the central node, the resulting graph is a cycle; if we remove any other node, the resulting graph  is an interval graph.
In both cases, we get a PCG \cite{Yal09,Bsurvey}.
\end{proof}

Now we prove that $C_n \square P_2$ is a minimal not PCG.
The proof is constructive and it provides an edge-weighted tree $T$ and two values $d_{min}$ and $d_{max}$ such that $PCG(T, d_{min}, d_{max})=C_n \square P_2 \setminus \{ x \}$ for any node $x$ of $C_n \square P_2$.

\begin{thm}
\label{th.minimal}
The graph obtained by removing any node from $C_n \square P_2$, $n\geq 4$, is PCG. In other words, $C_n \square P_2$ is a minimal not PCG.
\end{thm}

\begin{proof}
To prove the statement, we remove from the graph a node $x$ and prove that the new graph $G'$ is PCG. In view of the symmetry of the graph, it is not restrictive to assume that $x=u_n$. 
We construct a tree $T$ such that $G'=$ PCG($T$, $2n-2$, $2n+2$).  

We distinguish the following two cases depending on whether $n$ is an even or an odd number:
\begin{description}
\item[$n$ is an even number:] (refer to Figure \ref{fig.minimal}.a) tree $T$ is a caterpillar with  $n-1$ internal nodes that we denote as $x_1,x_2,\ldots, x_{n-1}$. The internal nodes induce a path from $x_1$ to $x_{n-1}$ and edges on this path $(x_i,x_{i+1})$, $1\leq i<n-1$, have all weight $2$. 
Leaves $v_i$ and $u_i$, $1\leq i<n$,  are connected  to $x_i$ with edges of weight $n$. Finally leaf $v_n$ is connected to the node $x_{\frac{n}{2}}$ with an edge of weight $1$. 

\item[$n$ is an odd number:] (refer to Figure \ref{fig.minimal}.b) tree $T$ is a caterpillar with  $n$ internal nodes that we denote as $x_1,x_2,\ldots, x_{n-1}$ and $y$. The internal nodes $x_1,\ldots, x_{\left\lfloor\frac{n}{2}\right\rfloor}$ induce a path from $x_1$ to $ x_{\left\lfloor\frac{n}{2}\right\rfloor}$ and edges $(x_i,x_{i+1})$, $1\leq i<\left\lfloor\frac{n}{2}\right\rfloor$, have weight $2$. The internal nodes $ x_{\left\lceil\frac{n}{2}\right\rceil}, \ldots, x_{n-1}$ induce a path from 
$ x_{\left\lceil\frac{n}{2}\right\rceil},$ to $x_{n-1}$  and edges $(x_i,x_{i+1})$, $\left\lceil\frac{n}{2}\right\rceil \leq i<n-1$, have weight $2$. Leaves $v_i$ and $u_i$, $1\leq i<n$,  are connected  to $x_i$ with edges of weight $n$.  Finally the internal node $y$ is connected to $ x_{\left\lfloor\frac{n}{2}\right\rfloor}$, $ x_{\left\lceil\frac{n}{2}\right\rceil}$ and $v_n$ with edges of weight $1$. 
\end{description}
In both cases, $G'= PCG(T, 2n-2, 2n+2)$.
\end{proof}

\begin{figure}[t]
\begin{tikzpicture}
\draw[fill=black] (0,1) circle (0.1cm); \node[] at (0,0.7) {$u_1$};
\draw[fill=black] (0.5,1) circle (0.1cm); \node[] at (0.5,0.7) {$v_1$};
\draw[fill=black] (0.25,2) circle (0.1cm); \node[] at (0.25,2.25) {$x_1$};
\draw[] (0,1) -- (0.25,2);
\draw[] (0.5,1) -- (0.25,2);

\draw[fill=black] (1,1) circle (0.1cm); \node[] at (1,0.7) {$u_2$};
\draw[fill=black] (1.5,1) circle (0.1cm); \node[] at (1.5,0.7) {$v_2$};
\draw[fill=black] (1.25,2) circle (0.1cm); \node[] at (1.25,2.25) {$x_2$};
\draw[] (1,1) -- (1.25,2);
\draw[] (1.5,1) -- (1.25,2);

\draw[fill=black] (2.5,1) circle (0.1cm); \node[] at (2.5,0.7) {$u_{n \over 2}$};
\draw[fill=black] (3,1) circle (0.1cm); \node[] at (3,0.7) {$v_{n \over 2}$};
\draw[fill=black] (2.75,2) circle (0.1cm); \node[] at (2.55,2.25) {$x_{n \over 2}$};
\draw[fill=black] (2.75,3) circle (0.1cm); \node[] at (2.75,3.25) {$v_{n}$};
\draw[] (2.5,1) -- (2.75,2);
\draw[] (3,1) -- (2.75,2);
\draw[] (2.75,2) -- (2.75,3);

\draw[fill=black] (4,1) circle (0.1cm); \node[] at (4,0.7) {$u_{n-1}$};
\draw[fill=black] (4.5,1) circle (0.1cm); \node[] at (4.7,0.7) {$v_{n-1}$};
\draw[fill=black] (4.25,2) circle (0.1cm); \node[] at (4.25,2.25) {$x_{n-1}$};
\draw[] (4,1) -- (4.25,2);
\draw[] (4.5,1) -- (4.25,2);

\draw[] (0.25,2) -- (1.25,2);
\draw[dashed] (1.25,2) -- (1.8,2);
\draw[dashed] (2.25,2) -- (3.25,2);
\draw[dashed] (3.75,2) -- (4.25,2);

\node[] at (0.7,2.2) {$2$};
\node[] at (0,1.4) {$n$};
\node[] at (0.5,1.4) {$n$};
\node[] at (1,1.4) {$n$};
\node[] at (1.5,1.4) {$n$};
\node[] at (2.5,1.4) {$n$};
\node[] at (3,1.4) {$n$};
\node[] at (4,1.4) {$n$};
\node[] at (4.5,1.4) {$n$};
\node[] at (2.9,2.6) {$1$};

\node[] at (2.5,0) {a.};

\draw[fill=black] (6,1) circle (0.1cm); \node[] at (6,0.7) {$u_1$};
\draw[fill=black] (6.5,1) circle (0.1cm); \node[] at (6.5,0.7) {$v_1$};
\draw[fill=black] (6.25,2) circle (0.1cm); \node[] at (6.25,2.25) {$x_1$};
\draw[] (6,1) -- (6.25,2);
\draw[] (6.5,1) -- (6.25,2);

\draw[fill=black] (7,1) circle (0.1cm); \node[] at (7,0.7) {$u_2$};
\draw[fill=black] (7.5,1) circle (0.1cm); \node[] at (7.5,0.7) {$v_2$};
\draw[fill=black] (7.25,2) circle (0.1cm); \node[] at (7.25,2.25) {$x_2$};
\draw[] (7,1) -- (7.25,2);
\draw[] (7.5,1) -- (7.25,2);

\draw[fill=black] (8.5,1) circle (0.1cm); \node[] at (8.4,0.7) {$u_{\lfloor {n \over 2} \rfloor}$};
\draw[fill=black] (9,1) circle (0.1cm); \node[] at (9.2,0.7) {$v_{\lfloor {n \over 2} \rfloor}$};
\draw[fill=black] (8.75,2) circle (0.1cm); \node[] at (8.55,2.25) {$x_{\lfloor {n \over 2} \rfloor}$};
\draw[] (8.5,1) -- (8.75,2);
\draw[] (9,1) -- (8.75,2);

\draw[fill=black] (9.5,2) circle (0.1cm); \node[] at (9.35,2.2) {$y$};
\draw[fill=black] (9.5,3) circle (0.1cm); \node[] at (9.25,3.25) {$v_{n}$};
\draw[] (9.5,2) -- (9.5,3);

\draw[fill=black] (10,1) circle (0.1cm); \node[] at (10,0.7) {$u_{\lceil {n \over 2} \rceil}$};
\draw[fill=black] (10.5,1) circle (0.1cm); \node[] at (10.7,0.7) {$v_{\lceil {n \over 2} \rceil}$};
\draw[fill=black] (10.25,2) circle (0.1cm); \node[] at (10.25,2.25) {$x_{\lceil {n \over 2} \rceil}$};
\draw[] (10,1) -- (10.25,2);
\draw[] (10.5,1) -- (10.25,2);

\draw[fill=black] (11.5,1) circle (0.1cm); \node[] at (11.5,0.7) {$u_{n-1}$};
\draw[fill=black] (12,1) circle (0.1cm); \node[] at (12.2,0.7) {$v_{n-1}$};
\draw[fill=black] (11.75,2) circle (0.1cm); \node[] at (11.75,2.25) {$x_{n-1}$};
\draw[] (11.5,1) -- (11.75,2);
\draw[] (12,1) -- (11.75,2);

\draw[] (6.25,2) -- (7.25,2);
\draw[] (8.75,2) -- (10.25,2);

\draw[dashed] (7.25,2) -- (7.8,2);
\draw[dashed] (8.25,2) -- (8.55,2);
\draw[dashed] (10.25,2) -- (10.75,2);
\draw[dashed] (11.25,2) -- (11.75,2);

\node[] at (6.7,2.2) {$2$};
\node[] at (6,1.4) {$n$};
\node[] at (6.5,1.4) {$n$};
\node[] at (7,1.4) {$n$};
\node[] at (7.5,1.4) {$n$};
\node[] at (8.5,1.4) {$n$};
\node[] at (9,1.4) {$n$};
\node[] at (10,1.4) {$n$};
\node[] at (10.5,1.4) {$n$};
\node[] at (11.5,1.4) {$n$};
\node[] at (12,1.4) {$n$};

\node[] at (9.4,2.6) {$1$};

\node[] at (9.1,2.2) {$1$};
\node[] at (9.7,2.2) {$1$};

\node[] at (9.5,0) {b.};
\end{tikzpicture}
\vspace*{-0.5cm}
\caption{Caterpillars for the proof of Theorem \ref{th.minimal}.}
\label{fig.minimal}
\end{figure}
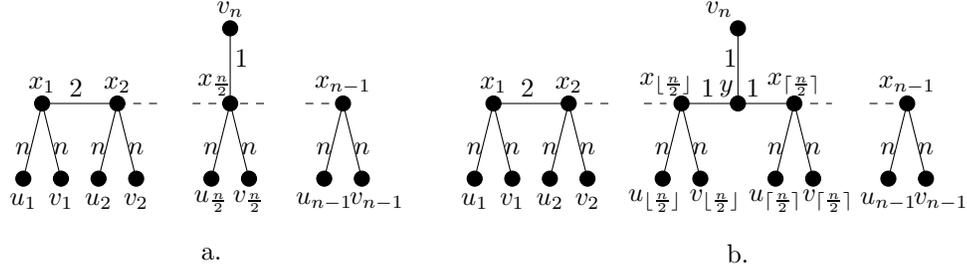

\section{Conclusions}
\label{sec.conclusion}

In this paper we proposed a new proof technique to show that graphs are not PCGs. As an example, we applied it to wheels and to $C_n \square P_2$.
Note that both these two classes are obtained by operating on cycles and recently we have used the same approach to prove that also the square of an $n$ node cycle, $n \geq 8$, is not PCG. 
 
Nevertheless, we think that this technique can be potentially used to collocate  outside PCG many other graph classes no related to cycles.
This represents an important step toward the solution of the very general open problem consisting in demarcating the boundary of the PCG class.

\bibliography{references}
\end{document}